\documentclass[a4paper,11pt]{scrartcl}
\addtokomafont{disposition}{\rmfamily}

\usepackage[left=2.5cm,right=2.5cm,top=2.5cm,bottom=2.5cm,includefoot]{geometry}

\usepackage{amsthm}
\usepackage{amsmath}
\usepackage{amssymb}
\usepackage{mathrsfs} 
\usepackage{graphicx}
\usepackage{appendix}

\usepackage[%
natbibapa
]{apacite}

\usepackage[dvipsnames]{xcolor}
\usepackage{stackengine}
\usepackage{placeins}
\usepackage{arydshln}
\usepackage[colorlinks=true,linkcolor=blue, citecolor=blue, urlcolor=blue, linktocpage=true]{hyperref} 
\usepackage{bbm}
\usepackage{textcomp}
\usepackage{enumerate}
\usepackage{mathtools}
\usepackage{tabularx}
\usepackage{booktabs}
\usepackage{rotating}

\usepackage{color}
\usepackage{framed}
\usepackage{comment}
\definecolor{shadecolor}{gray}{0.9}

\usepackage{dsfont} 
\usepackage{caption}
\usepackage{subcaption}
\usepackage{graphicx}
\usepackage{pdfpages}
\usepackage{url}
\usepackage[english]{babel}
\usepackage{setspace}
\usepackage{dashrule}
\usepackage{tikz}
\usepackage{bm}

\usetikzlibrary{automata, positioning, arrows}
\tikzset{
	->,  
	node distance=5.5cm, 
	every state/.style={thick, fill=gray!10}, 
	initial text=$ $, 
}

\specialcomment{extra}{\begin{shaded}}{\end{shaded}}

\theoremstyle{plain} 
\newtheorem{theorem}{Theorem}[section] 
\newtheorem{lemma}[theorem]{Lemma} 
\newtheorem{proposition}[theorem]{Proposition} 
\newtheorem{corollary}[theorem]{Corollary} 

\theoremstyle{definition} 
\newtheorem{definition}[theorem]{Definition}
 
\newtheorem{example}[theorem]{Example}
\newtheorem{remark}[theorem]{Remark}
\newtheorem{assump}[theorem]{Assumption}

\newcommand{\p}{\mathbb{P}}
\newcommand{\E}{\mathbb{E}}
\renewcommand{\P}{\mathbb{P}}
\newcommand{\Var}{\mathrm{Var}}

\newcommand{\Cov}{\mathrm{Cov}} 
\newcommand{\Cor}{\mathrm{Cor}}
\newcommand{\QCov}{\mathrm{QCov}}
\newcommand{\QCor}{\mathrm{QCor}}
\newcommand{\TCov}{\mathrm{TCov}}
\newcommand{\TCor}{\mathrm{TCor}}

\newcommand{\SCov}{\mathrm{SCov}}
\newcommand{\SCor}{\mathrm{SCor}}
\newcommand{\ECov}{\mathrm{ECov}} 
\newcommand{\ECor}{\mathrm{ECor}}
\newcommand{\LTCor}{\mathrm{LTCor}}
\newcommand{\UTCor}{\mathrm{UTCor}}
\newcommand{\MCor}{\mathrm{MCor}}
\newcommand{\QFCor}{\mathrm{QFCor}}
\newcommand{\CDFCor}{\mathrm{CDFCor}}
\newcommand{\QFCov}{\mathrm{QFCov}}
\newcommand{\CDFCov}{\mathrm{CDFCov}}
\newcommand{\CDF}{\mathrm{CDF}}
\newcommand{\QF}{\mathrm{QF}}
\newcommand{\DCor}{\mathrm{DCor}}
\newcommand{\dist}{\stackrel{\mathrm{d}}{=}}
\newcommand{\ES}{\mathrm{ES}}

\newcommand{\D}{\mathcal{D}}
\renewcommand{\d}{\,\mathrm{d}}
\newcommand{\R}{\mathbb{R}}
\newcommand{\T}{T}
\newcommand{\one}{\mathds{1}}
\newcommand{\q}{q}
\newcommand{\eps}{\varepsilon}

\newcommand{\A}{\mathsf{A}}

\newcommand{\essinf}{\mathrm{ess\,inf}}
\newcommand{\esssup}{\mathrm{ess\,sup}}
\renewcommand{\L}{\mathcal{L}}

\def\be{\begin{equation} \label}
	\def\ee{\end{equation}}

\usepackage[colorinlistoftodos,textsize=tiny]{todonotes}
\newcommand{\Comments}{1}
\newcommand{\mynote}[2]{\ifnum\Comments=1\textcolor{#1}{#2}\fi}
\newcommand{\mytodo}[2]{\ifnum\Comments=1%
	\todo[linecolor=#1!80!black,backgroundcolor=#1,bordercolor=#1!80!black]{#2}\fi}

\ifnum\Comments=1               
\fi

\ifnum\Comments=1               
\fi

\numberwithin{equation}{section} 

\graphicspath{{./Plots/}}

\begin{document}
	
\title{Generalised Covariances and Correlations}

\author{Tobias Fissler\thanks{
		RiskLab, Department of Mathematics, ETH Zurich, R\"amistrasse 101, 8092 Zurich, Switzerland, \newline
		e-mail: \href{mailto:tobias.fissler@math.ethz.ch}{tobias.fissler@math.ethz.ch} 
	}\and Marc-Oliver Pohle\thanks{Schumpeter School of Business and Economics, University of Wuppertal, Gau\ss stra\ss e 20, 42119 Wuppertal, Germany, and Heidelberg Institute for Theoretical Studies, 
	Heidelberg, Germany, \newline e-mail: \href{mailto: pohle@uni-wuppertal.de}{pohle@uni-wuppertal.de}}}

\maketitle

\begin{abstract}
The covariance of two random variables measures the average joint deviations from their respective means. 
We generalise this well-known measure by replacing the means with other statistical functionals such as quantiles, expectiles, or thresholds. 
Deviations from these functionals are defined via generalised errors, typically induced through identification or moment functions. 
As a normalised measure of dependence, a generalised correlation is constructed. 
Replacing the common Cauchy--Schwarz normalisation by a novel Fr\'echet--Hoeffding normalisation, we obtain attainability of the entire interval $[-1, 1]$  for any given marginal distributions. 
We uncover favourable properties of these new dependence measures and establish consistent estimators. 
The families of quantile and threshold correlations make it possible to measure local dependence and give rise to function-valued distributional correlations, exhibiting the entire dependence structure. 
Quantile correlations also lead to tail correlations, new measures of tail dependence, closely related to and refining classical coefficients of tail dependence.
Finally, we construct summary covariances (correlations), a class of regional or global dependence measures, which arise as (normalised) weighted averages of distributional covariances.  
We retrieve covariance, Pearson and Spearman correlation as special cases. The usefulness of our new dependence measures is illustrated on demographic data from the Panel Study of Income Dynamics.
\end{abstract}

\noindent
\textit{Keywords:}
dependence measure;
statistical functional;
identification function;
quantile correlation;
copula;
tail dependence\\


\section{Introduction}

Measuring the dependence of two random variables $X$ and $Y$ has been a long-standing task in statistics with relevance for almost any empirical field of science. The two key approaches to this task are regression analysis, considering $Y$ conditional on $X$, and mutual dependence measures. The most popular measures of dependence are \emph{covariance}, \emph{Pearson correlation} and the rank correlation coefficients \emph{Spearman's $\rho$} and \emph{Kendall's $\tau$}. 
Overviews of the vast literature on dependence measures are given in \cite{Mari2001}, \cite{Balakrishnan2009}, and \cite{Tjostheim2022}. Directed dependence measures are usually normalised to range between $-1$ and $1$.
They indicate the direction of dependence by their sign and the strength of dependence by the proximity of their absolute value to 1. 
Crucial for the usefulness and interpretability of dependence measures are certain properties,  
originating from
R\'enyi's axioms  \citep{Renyi1959} and often modified subsequently \citep{Schweizer1981, Embrechts2002, Balakrishnan2009}. 
In particular, a dependence measure should indicate the 
special situations of independence, perfect positive and perfect negative dependence
by attaining the values 0, 1 and $-1$, respectively. 
Pearson correlation suffers from some well-known shortcomings \citep{Embrechts2002}, most importantly, attainability issues: 
In general,
for given marginal distributions of $X$ and $Y$, 
the Pearson correlation only ranges in a subinterval of $[-1,1]$.
This seriously impacts the interpretability of Pearson correlation.

Let us now recall the definitions of \emph{covariance} $\Cov(X,Y)$ and \emph{Pearson correlation} $r(X,Y)$: 
\begin{equation}
\label{eq:Pearson}
\Cov(X,Y) = \E\big[(X-\mu(X))(Y-\mu(Y))\big], \qquad r(X,Y) = \frac{\Cov(X,Y)}{\sqrt{\Var(X) \Var(Y)}},
\end{equation}
where $\Var(X) = \Cov(X,X)$ denotes the variance of $X$. Covariance measures the average joint deviations of $X$ and $Y$ around their respective means, $\mu(X)$ and $\mu(Y)$. 
Pearson correlation is a normalised version of covariance, exploiting the Cauchy--Schwarz inequality.

\subsection{Contributions and plan of the paper}
 
We aim at measuring the average joint deviations of $X$ and $Y$ around general \emph{statistical functionals} $T_1(X)$ and $T_2(Y)$ such as quantiles, expectiles or thresholds, not only the mean.
This provides a more complete picture of the dependence structure. 
Widening the perspective along these lines
is akin to the methodological advancement in other key branches of statistics:
In univariate statistics, statistical functionals constitute valuable summary measures complementing the mean.
In regression analysis, the shift to modelling functionals of the conditional distribution other than the mean was initiated by the advent of quantile and expectile regression \citep{koenker1978, newey1987}.

In Section \ref{sec:generalised}, we introduce generalised covariance. Replacing the mean by another functional necessitates a different, in general non-linear measurement of deviation from the functional of interest to preserve the property that independence implies nullity of the dependence measure.
We replace the classical error, $X-\mu(X)$, by generalised errors, typically constructed via identification functions. 
Then, our new generalised covariances are defined as the expectation of the product of these generalised errors.

Next, we normalise the generalised covariance accordingly to arrive at our generalised correlation (Section \ref{sec:generalised correlations}). We show that the classical Cauchy--Schwarz normalisation employed in Pearson correlation should not be used here as it leads to serious attainability issues. Instead, we propose an alternative and natural normalisation via Fr\'echet--Hoeffding bounds, which are sharp by construction. 
This novel normalisation can also straightforwardly be used for classical covariance, which leads to what we call \textit{mean correlation}, an attainable version of Pearson correlation. 
The generalised correlations have favourable properties, allow for measuring new forms of dependence and thus for gaining additional insights about the dependence structure between $X$ and $Y$. 
In particular, \textit{quantile correlation} and the closely related \textit{threshold correlation} arise, which measure dependence locally around a pair of quantiles of $X$ and $Y$ or around any point in the codomain of $(X,Y)$. 
Quantile correlation can be regarded as the correlation analogue to quantile regression.  
Further examples of generalised correlations are expectile correlation and mean-quantile correlation, where the first statistical functional $T_1(X)$ is the mean and the second one $T_2(Y)$ is a quantile. 

Considering all the information on local dependence jointly, that is, using the whole families of quantile or of threshold covariances and correlations leads to what we call quantile function and CDF (cumulative distribution function) covariances or correlations, respectively, which we subsume under the term \textit{distributional covariances} and \textit{correlations} (Section \ref{sec:local_distributional}). They are function-valued objects revealing the entire dependence structure, with one-to-one connections to the copula (for the former) and the joint distribution of $X$ and $Y$ (for the latter). 
Since they are normalised, attaining values in $[-1,1]$, their interpretation is a lot easier than interpreting copulas or joint distribution functions. 
Interestingly, a constant zero (one\,/\,negative one) of these distributional correlations implies independence (perfect positive\,/\,perfect negative dependence) of $X$ and $Y$. 
The distributional correlations can be interpreted as generalised correlations 
of the identity functional, that is, of the full distributions itself.
This fact justifies our nomenclature and shows that distributional correlations are the counterpart of distributional regression \citep{chernozhukov2013, kneib2021}.

\begin{figure}
	\centering
	\includegraphics[width=1\linewidth]{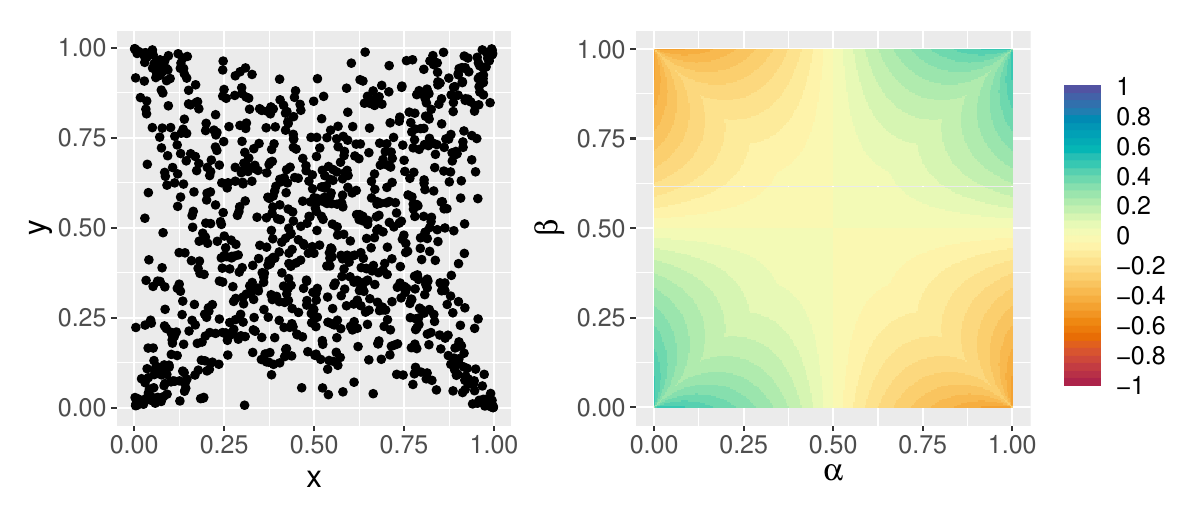}
	\caption{Plot of quantile function correlation, $\QFCor$, for a bivariate Cauchy copula with Spearman's $\rho$ equal to 0 (right) and a scatter plot of 1000 random draws from it (left)}
	\label{fig:qfcorcauchy}
\end{figure}

Figure \ref{fig:qfcorcauchy} depicts a first example of a distributional correlation. It presents quantile function correlation (i.e.\ quantile correlation as a function of the quantile levels $\alpha$ and $\beta$, where $T_1(X)$ is the $\alpha$-quantile of $X$ and $T_2(Y)$ is the $\beta$-quantile of $Y$) of a bivariate Cauchy copula with a Spearman correlation of 0 alongside a scatter plot of 1000 random draws from this copula. 
Despite the Spearman correlation of 0, the quantile function correlation reveals 
a quite interesting dependence structure.  
In the upper right and lower left, the dependence is positive, meaning that larger values (exceedances of certain quantiles) of $X$ are associated with larger values (exceedances of certain quantiles) of $Y$, while in the upper left and lower right it is negative, meaning that larger values (exceedances of certain quantiles) of $X$ are associated with smaller values (falling short of certain quantiles) of $Y$. Furthermore, the dependence gets stronger closer to the tails.

In Section \ref{sec:tail}, we introduce \textit{tail correlations}, which arise as limits of quantile correlations, 
and we elaborate on their connection to the widely-used coefficients of tail dependence \citep{Coles1999, Joe2014}. 
Since they are derived from correlations, tail correlations attain values between $-1$ and $1$. 
Our new measures coincide with the coefficients of tail dependence if they are positive, yet nicely distinguish between different strengths of negative tail dependence and tail independence, whereas coefficients of tail dependence are 0 in these cases. In the example illustrated in Figure \ref{fig:qfcorcauchy}, the upper and lower tail correlations arise when moving along the diagonal to the upper right and lower left corners of the plot (letting the quantile levels $\alpha$ and $\beta$ simultaneously go to 1 or 0, respectively). 
In the concrete case, they equal 0.292, coinciding with the corresponding coefficients of tail dependence since the tail dependence is positive.

Section \ref{sec:summary} reconsiders the classical task of dependence measures, condensing overall dependence into a single number, and elaborates on the natural idea of averaging the distributional covariances and correlations to summarise the full information on the dependence structure contained in them. 
Our \textit{summary covariances} are constructed as weighted averages (integrals) of distributional covariances, and corresponding \textit{summary correlations} arise by properly normalising them. 
Strikingly, this recovers classical and Spearman covariance as well as mean and Spearman correlation as canonical special cases (under equal weighting). 
In the example of Figure \ref{fig:qfcorcauchy}, the canonical summary correlation for quantile function correlation, which corresponds to Spearman's $\rho$, is 0 due to cancellation effects.

The development of new statistical methods usually proceeds in two steps. Firstly, useful population concepts are devised as a foundation. Secondly, estimation and inference need to be addressed. 
The focus of this paper lies on the first step, while still proposing estimators and providing consistency results for them (Section \ref{subsec:estimation}).
A treatment of the limiting distributions of the estimators of our various measures is beyond the scope of this paper, but we provide directions for future research regarding uncertainty quantification.
 
In Section \ref{subsec:data examples}, we illustrate the usage of our newly introduced dependence measures on demographic data stemming from the Panel Study of Income Dynamics.
The Appendix contains proofs and additional details on the relation of generalised errors and identification functions and additional material on the theoretical examples and the empirical applications. We provide an R package accompanying this paper at \url{https://github.com/MarcPohle/GCor}.

\subsection{Related literature}

While we are not aware of any attempts to generalise covariance and Pearson correlation to measure dependence around arbitrary statistical functionals, several measures exist related to generalised covariances and correlations for specific functionals, in particular quantiles.
Blomqvist's $\beta$ \citep{Blomqvist1950} coincides with our median correlation in the continuous case. Furthermore, \cite{Linton2007} employ a closely related quantity in their popular quantilogram, namely our quantile covariance (see Example \ref{exmp:quantile cov}), but normalised with the Cauchy--Schwarz normalisation, see also \cite{Han2016}. Similarly, what \cite{Li2015} call quantile correlation amounts to our quantile-mean covariance (see Example \ref{exmp:quantile-mean}), but again normalised with the Cauchy--Schwarz normalisation. 
These measures exhibit severe attainability issues that go hand in hand with the Cauchy--Schwarz normalisation, see our detailed discussion in Section \ref{sec:generalised correlations}. Further, so-called indicator covariances from spatial statistics \citep{Dubrule2017} are closely related to our threshold covariances (see Example \ref{example:Threshold covariance}). 

Our local (threshold and quantile) covariances and correlations can be regarded as local measures of dependence around a certain point, while our distributional correlations consider the entire local information jointly, uncovering the full dependence structure. Surprisingly, such local and full dependence measures have hardly been touched upon in the literature. Noteworthy exceptions are the local dependence function of \cite{Holland1987}, see also \cite{Jones1996}, and the local Gaussian correlation of \cite{Tjostheim2013}. However, these measures lack important theoretical properties and their estimation can pose difficulties.

\section{Generalised covariances}
\label{sec:generalised}

\subsection{Generalised errors}

Let $(\Omega, \mathfrak{F}, \mathbb P)$ be a non-atomic probability space. Denote by $L^0(\R^d)$, $d=1,2$, the space of all $\R^d$-valued random variables. 
For $p\in[1,\infty)$, let $L^p(\R) = \{X\in L^0(\R) \mid \E[|X|^p]<\infty\}$ and $L^p(\R^2) = \{(X,Y) \in L^0(\R^2) \mid X, Y \in L^p(\R)\}$. 
We consider statistical functionals $\T$ as law-determined maps from some collection of random variables $\mathcal L\subseteq L^0(\R)$ to a set $\A\subseteq \R$, 
meaning that for any $X,X'\in\L$ it holds that $\T(X) = \T(X')$ whenever $F_X=F_{X'}$, where $F_X,F_{X'}$ are the distribution functions of $X,X'$, respectively.

Generalised covariances and generalised correlations are law-determined maps from a class $\D \subseteq L^0(\R^2)$ of bivariate random vectors $(X,Y)$ to $\R$ and to $[-1,1]$, respectively.
Recall the classical covariance and Pearson correlation on  
$\D = L^2(\R^2)$ from \eqref{eq:Pearson}. 
The rationale behind the definition of covariance is to measure average co-movements of $X$ and $Y$ \emph{around their respective means}, $\mu(X)$ and $\mu(Y)$.
The very idea behind generalised covariances is to measure average co-movements around functionals $T_1$ and $T_2$ other than the mean, e.g., around certain quantiles of $X$ and $Y$. 
A naive ansatz is to merely replace $\mu$ by $\T$ in the definition of the covariance. 
However, this is not a suitable way to measure co-movements around arbitrary functionals. For example the fundamental property that independence of $X$ and $Y$ implies nullity of the generalised covariance would be violated. 
Covariance is constructed via deviations from the means of $X$ and $Y$, or \emph{errors}, $X-\mu(X)$ and $Y-\mu(Y)$. We need to find a suitable way to measure deviations of a random variable $X$ from an arbitrary functional $\T(X)$, which leads to the notion of \emph{generalised errors for $\T$}, capturing the most important properties of the prototypical error $X-\mu(X)$: having mean zero, being positive (negative) if $X$ realises above (below) $T(X)$, and being (weakly) larger the further $X$ realises away from $T(X)$.

\begin{definition}[Generalised error]
	\label{def:generalised error}
	For a given functional $T\colon\L\to\A\subseteq\R$, we call a map $e_T\colon\L\to L^1(\R)$ a generalised error for $T$ if the following properties hold for all $X\in\L$.
	\begin{enumerate}[(i)]
		\item
		\label{prop:centred}
		Centred: $\E[e_{T}(X)]=0$.
		\item
		\label{prop:increasing}
		Increasing: for $\P\otimes \P$-almost all $(\omega, \omega')\in\Omega^2$ 
		\[
		X(\omega)\ge X(\omega') \implies e_{T}(X)(\omega)\ge e_{T}(X)(\omega').
		\]
		\item
		\label{prop:sign change}
		Sign change at $\T$: for all $X\in\mathcal L$ and for $\P$-almost all $\omega\in\Omega$
		\begin{equation}
			\label{eq:sign change}
			\big(X(\omega) - \T(X)\big) e_{T}(X)(\omega)\ge0.
		\end{equation}
	\end{enumerate}
\end{definition}

Further examples of generalised errors besides the prototypical $e_\mu(X) = X - \mu(X)$ are discussed in Subsection \ref{subsec:Examples of generalised covariances}. 
A natural way to construct a generalised error map is via so-called \emph{identification functions}.

\begin{definition}[Identification function]
	Let $\A\subseteq \R$.
	A map $v\colon \A\times\R\to\R$ is called $\L$-integrable if for all $t\in\A$ 
	and $X\in\L$ it holds that $\E|v(t,X)|<\infty$.
	Moreover, $v$ is called increasing\,/\,non-constant 
	if for any $t\in\A$ the map $x\mapsto v(t,x)$ is increasing\,/\,non-constant.

	An $\L$-integrable map $v\colon \A\times\R\to\R$ is an $\L$-identification function for a functional $\T\colon\L\to\A\subseteq \R$ if 
	$\E\big[v(\T(X),X)\big]=0$ for all $X\in\L$.
	It is a \emph{strict} $\L$-identification function if additionally 
	\[
	\E\big[v(t,X)\big]=0 \implies t=\T(X)
	\]
	for all $t\in\A$ and for all $X\in\L$. 
	$\T$ is identifiable on $\L$ if there exists a strict $\L$-identification function for it.
\end{definition}

In the field of forecast evaluation, identification functions are a central tool to assess forecast calibration \citep{NoldeZiegel2017,DimiPattonSchmidt2019}. 
In econometrics, they are often known as moment functions, and they give rise to Z-estimation or the (generalised) method of moments estimation \citep{Huber1967, Hansen1982, NeweyMcFadden1994}. 
An example for a strict $L^1(\R)$-identification function for the mean, which induces the error for the mean $e_\mu(X)$ from above, is  $v_\mu(t,x) = x-t$ (see again Subsection \ref{subsec:Examples of generalised covariances} for further examples).
The following proposition provides a recipe how to build generalised errors from identification functions, which will be the construction principle for all generalised errors in this paper but the ones for thresholds (Example \ref{example:Threshold covariance}) and quantiles in the non-continuous case (Example \ref{example:discontinuous marginals}), where there is still a very close connection to identification functions. We need the following assumption on $T$ and especially on the class $\L$ where it is defined.

	\begin{assump}\label{ass:T}
		Let $\T: \L \to\A\subseteq \R$ be a functional. 
		\begin{enumerate}[(a)]
			\item
			For each $X\in \L$ it holds that 
			\be{eq:bounded}
			\essinf (X)\le  T(X)\le \esssup (X).
			\ee
			\item
			For each $(t,x)\in\A\times \R$, $x\neq t$, there is an $X\in\L$ with $t=\T(X)$ such that
			$x = \esssup(X)$ if $t<x$, and
			$x = \essinf(X)$ if $t>x$.
		\end{enumerate}
	\end{assump}

    Standard examples for $\T$ such as the mean or a quantile fulfil Assumption \ref{ass:T} (a). 
	For most other practically relevant and identifiable functionals, one can achieve \eqref{eq:bounded} by applying a map $m\colon\R\to\R$ to $X$, which results in $T\big( m(\cdot)\big)$, i.e., by considering a transformed random variable.
	E.g., the second moment can be achieved when $T$ is the mean functional and $m(x) = x^2$. 

	\begin{proposition}
		\label{prop:generalised error}
		Let $v_T:\A\times \R\to\R$ be an increasing, non-constant $\L$-identification function for the functional $T:\L\to\A\subseteq \R$ satisfying Assumption \ref{ass:T}. Then, for a random variable $X\in\L$, the quantity 
		\be{eq:error}
		\omega\mapsto e_{v_T}(X)(\omega):= v_T\big(T(X),X(\omega)\big),
		\ee
		is a generalised error of $X$ for $T$.
	\end{proposition}

Usually, an estimator or forecast for the functional $T$ is plugged in as the first argument of the identification function. This proposition shows that simply plugging in the true functional itself yields a suitable way to measure the deviation of a random variable $X$ from a functional $T(X)$ of itself.

\subsection{Definition and properties}

\begin{definition}[Generalised covariance] \label{defn:generalised covariance}
Let $T_1\colon\L_1\to\A_1\subseteq\R$, $T_2\colon\L_2\to\A_2\subseteq\R$ be two functionals and $e_{T_1}$ and $e_{T_2}$ generalised errors for these functionals.
Let $\D = \{(X,Y) \in L^0(\R^2)\,|\,X \in\L_1,\ Y \in \L_2,\ e_{T_1}(X) e_{T_2}(Y) \in L^1(\R)\}$.
Then, the \emph{generalised covariance at $\T_1$ and $\T_2$ induced by $e_{T_1}$ and $e_{T_2}$}, or the \emph{$\T_1-\T_2$-covariance induced by $e_{T_1}$ and $e_{T_2}$}, 
is defined on $\D$ via
\begin{equation}
\label{eq:generalised cov}
\Cov_{\T_1, \T_2}(X,Y) := \E\big[e_{T_1}(X) e_{T_2}(Y)\big].
\end{equation}
\end{definition}
A classical sufficient condition for the integrability of the product $e_{T_1}(X) e_{T_2}(Y)$ is that $e_{T_1}(X)\in L^p(\R)$ and $e_{T_2}(Y)\in L^q(\R)$ for some $p,q\in[1,\infty]$ with $p^{-1}+q^{-1} = 1$, exploiting H\"older's inequality. 
A special case of this is the Cauchy--Schwarz inequality. 
An alternative and weaker condition is provided in Proposition \ref{prop:Frechet-Hoeffding}.

Since the generalised errors are centred by definition, $\E\big[e_{T_1}(X)\big] = \E\big[e_{T_2}(Y)\big]=0$, independence implies nullity:
\begin{proposition}
For the generalised covariance from \eqref{eq:generalised cov} it holds that $\Cov_{\T_1, \T_2}(X,Y) = 0$ if $X$ and $Y$ are independent.
\end{proposition}

\begin{remark}
The fact that the error terms are centred also implies that the generalised covariance can equivalently be written as the covariance of the generalised errors. That is,
\begin{equation}
\label{eq:cov identity}
\Cov_{\T_1, \T_2}(X,Y) = \Cov\big(e_{T_1}(X), e_{T_2}(Y)\big).
\end{equation}
\end{remark}

This nicely illustrates the rationale of generalised covariances at $\T_1$ and $\T_2$, measuring average co-movements of $X$ and $Y$ around their respective reference functionals.
An increasing likelihood of joint positive deviations of $X$ and $Y$ from $\T_1(X)$ and $\T_2(Y)$ leads to an increase in $\Cov_{\T_1, \T_2}(X,Y)$. 
On the other hand, an increasing likelihood of countermovements decreases this covariance.
Moreover, thanks to the errors being increasing, the value of the covariance is also sensitive to the magnitude of the deviations of $X$ and $Y$ from their reference functionals. 
Finally, if there is no systematic mutual influence between $X$ and $Y$, i.e., they are independent, the covariance vanishes.

Obviously, $\Cov_{\T_1, \T_2}(X,Y)$ depends on the choice of the generalised errors, $e_{T_1}$, $e_{T_2}$.
For the leading situation when the generalised error is induced by an identification function (see Proposition \ref{prop:generalised error}), we characterise this dependence in Section \ref{app:Generalised errors and identification functions} of the Appendix (Remark \ref{rem:dependence of id}) and remark that generalised correlations are actually independent of the choice of the identification function, subject to regularity conditions (Proposition \ref{prop:ind from h}).
For the examples discussed in the following subsection we utilise the canonical identification functions as suggested by \cite{GneitingResin2021}.

\subsection{Examples}
\label{subsec:Examples of generalised covariances}

\begin{example}[Classical covariance]
The mean has an increasing, non-con\-stant strict $L^1(\R)$-identification function
\(
v_\mu(t,x) := x-t, \ x,t\in\R.
\)
The induced error \eqref{eq:error} leads to the classical covariance \eqref{eq:Pearson}.
\end{example}

\begin{example}[Expectile covariance] \label{ex:Expectile_covariance}
The asymmetric version of the mean $\mu$, the $\tau$-expectile $\mu_\tau$, admits an increasing, non-constant strict $L^1(\R)$-identification function 
\begin{equation}
	\label{eq:expectile id}
	v_{\mu_\tau}(t,x) := 2|\one\{x\le t\} - \tau|(x-t), \qquad x,t\in\R,
\end{equation}
where $\tau\in(0,1)$. 
Clearly, for $\tau = 1/2$, this recovers the case of the mean.
The induced expectile covariance at levels $\tau,\eta\in(0,1)$ is 
\begin{align}
	\label{eq:expectile cov}
	\ECov_{\tau,\eta}(X,Y) 
	&:= \Cov_{\mu_{\tau}, \mu_{\eta}}(X,Y)\\ \nonumber
	&= 4 \E\big[|\one\{X\le \mu_{\tau}(X)\} - \tau|(X-\mu_{\tau}(X)) |\one\{Y\le \mu_{\eta}(Y)\} - \eta|(Y-\mu_{\eta}(Y))  \big],
\end{align}
where $X,Y, XY \in L^1(\R)$.
Clearly, for $\tau = \eta = 1/2$, one recovers the usual covariance.
\end{example}
Just as the mean, the $\tau$-expectile is translation equivariant in the sense that $\mu_\tau(X + c) = \mu_\tau(X)+c$ for all $X\in L^1(\R)$ and $c\in\R$, and it is positively homogeneous, i.e., $\mu_\tau(\lambda X) = \lambda \mu_\tau(X)$ for all $X\in L^1(\R)$ and $\lambda>0$.
The canonical expectile identification function \eqref{eq:expectile id} shares similar properties: It is positively homogeneous and translation \emph{in}variant. Hence, the expectile covariance is also translation invariant and positively homogeneous in both arguments. 
\begin{proposition}
\label{prop:properties_expectile_cor}
For all $\tau,\eta\in(0,1)$, for all $X,Y\in L^1(\R)$ such that $XY\in L^1(\R)$, for all $c\in\R$ and $\lambda>0$ it holds that
\begin{align*}
	\ECov_{\tau,\eta}(X + c,Y) &= \ECov_{\tau,\eta}(X,Y + c) = \ECov_{\tau,\eta}(X,Y),\\
	\ECov_{\tau,\eta}(\lambda X,Y) &=\ECov_{\tau,\eta}( X,\lambda Y) 
	= \lambda \ECov_{\tau,\eta}( X,Y) .
\end{align*}
\end{proposition}

\begin{example}[Threshold covariance]
\label{example:Threshold covariance}
The arguably simplest situation is to consider dependence around a point $(a,b) \in \mathbb{R}^2$, that is, to measure the average joint deviation of $X$ from an absolute threshold $a\in\R$ and of $Y$ from $b\in\R$.
This requires a special treatment since, formally, the functionals are constant. 
As such, they are identifiable, but the identification function for the constant $a$, $v_a(t,x) = t-a$, is constant in its second argument $x$.
Thus, it would only yield a trivial generalised error, which would be constant 0.
To circumvent this problem, consider
\be{eq:Threshold error}
e_a(X) = F_X(a) - \one\{X \le a\},
\ee
which is indeed for all $X\in L^0(\R)$ a generalised error for the constant functional $a$.
Hence, we can define the threshold covariance at points $a,b\in\R$ 
\begin{align}
\label{eq:Threshold cov}
\TCov_{a,b}(X,Y)
:= \E\big[\left( F_X(a) - \one \{ X \leq a\} \right) \left( F_Y(b) - \one \{ Y \leq b\} \right)\big] 
= F_{X,Y}(a,b) - F_X(a)F_Y(b).
\end{align}

Interestingly, the generalised error in \eqref{eq:Threshold error} can also arise via \eqref{eq:error} for the evaluation functional
$\T_a(F):= F(a)$
with the increasing, non-constant strict $L^0(\R)$-identification function
\begin{equation*}
v_{T_a} (t,x) := t - \one\{x \le a\}, \qquad x,t\in\R.
\end{equation*}
\end{example}

\begin{example}[Quantile covariance, continuous case]
\label{exmp:quantile cov}
For the (lower) $\alpha$-quantile given by $\q_\alpha(X) = \inf\{x\in\R\,|\, F_X(x)\ge\alpha\}$, $X\in L^0(\R)$, $\alpha\in(0,1)$, the function 
\begin{equation}
\label{eq:quantile id}
v_{\q_\alpha}(t,x) := \alpha - \one\{x\le t\}, \qquad x,t\in\R,
\end{equation}
is an increasing, non-constant $L_{\text{con}}$-identification function, where $L_{\text{con}} = \{X\in L^0(\R) \,|\,F_X \text{ is continuous}\}$. 
For $X,Y\in L_{\text{con}}$, the induced quantile covariance at levels $\alpha,\beta\in(0,1)$ is 
\begin{align}
\label{eq:quantile cov}
\QCov_{\alpha, \beta}(X,Y) 
:= \Cov_{\q_{\alpha}, \q_{\beta}}(X,Y)
&= \E\big[ (\alpha - \one \{X\leq \q_{\alpha}(X)\}) (\beta - \one \{Y\leq \q_{\beta}(Y)\})   \big] \\
\label{eq:quantile copula}
&= C_{X,Y}(\alpha,\beta) - \alpha\beta,
\end{align}
where $C_{X,Y}$ is the copula of $(X,Y)$ given by $C_{X,Y}(\alpha,\beta) = F_{X,Y} \left( \q_{\alpha}(X), \q_{\beta}(Y) \right)$, see \citet{sklar1959} and for an overview, e.g., \citet{Joe2014}.
Again, the quantile identification function \eqref{eq:quantile id} is bounded in $x$ such that we can dispense with integrability assumptions on $X$ and $Y$.
\end{example}

Quantile covariance enjoys a useful property, following from $q_{\alpha}(X)=-q_{1 - \alpha}(-X)$.

\begin{proposition}
	\label{properties_quantile_covariance_continuous}
	For all $\alpha,\beta\in(0,1)$, for all $X,Y\in L_{\text{con}}$ with $F_X$ and $F_Y$ strictly increasing on their support, it holds that
	\begin{align*}
		\QCov_{\alpha,\beta}(X,Y) &= -\QCov_{\alpha,1-\beta}(X,-Y) = -\QCov_{1-\alpha,\beta}(-X,Y).
	\end{align*}
\end{proposition}

\begin{example}[Quantile covariance, general case]
\label{example:discontinuous marginals}
On the entire $L^0(\R)$, $v_{\q_\alpha}$ defined in \eqref{eq:quantile id} fails to identify the $\alpha$-quantile, due to possible discontinuities in the CDF. 
However, a natural modification of the generalised error from the continuous case leads to a suitable error for the general case:
\be{eq:generalised error quantile}
e_{q_\alpha}(X) := F_X( q_\alpha (X) ) - \one\{X \le q_\alpha(X)\}.
\ee
This quantile error can be seen as a correction of the generalised error from the continuous case, replacing the quantile level $\alpha$ with the corrected quantile level $\alpha^*:=F_X(q_\alpha(X))$ 
accounting for a jump in the CDF and ensuring that 
$e_{q_\alpha}(X)$ is centred. On the other hand, it is just the natural analogue to the threshold error \eqref{eq:Threshold error}. 
It leads to the general definition of quantile covariance at levels $\alpha,\beta\in(0,1)$ as
\begin{align}\label{eq:Quantile correction}
	\QCov_{\alpha, \beta}(X,Y)
	= \E\big[ \big(F_X(q_\alpha(X)) - \one \{X\leq \q_{\alpha}(X)\}\big) \big(F_Y(q_\beta(Y)) - \one \{Y\leq \q_{\beta}(Y)\}\big)   \big].
\end{align}
Similar to \eqref{eq:quantile copula}, this general version of quantile covariance can also be expressed in terms of a copula of $C_{X,Y}$ of $(X,Y)$: 
\be{eq:Quantile correction copula}
\QCov_{\alpha, \beta}(X,Y)
= C_{X,Y}\big(F_X(q_\alpha(X)), F_Y(q_\beta(Y))\big) - F_X(q_\alpha(X)) F_Y(q_\beta(Y)).
\ee
Since all copulas for $(X,Y)$ coincide on $\text{range}(F_X) \times \text{range}(F_Y)$, the expression \eqref{eq:Quantile correction copula} is well defined, i.e., independent of the choice of the copula.
\end{example}

The quantile enjoys even more invariance properties than the expectile. It is equivariant under all strictly increasing left-continuous transformations: $q_\alpha\big(g(X)\big) = g\big(q_\alpha(X)\big)$ for any $X\in L^0(\R)$ and for any strictly increasing left-continuous $g\colon\R\to\R$.
The quantile error \eqref{eq:generalised error quantile} inherits this invariance, $e_{q_{\alpha}} (g(X)) = e_{q_{\alpha}} (X)$. Hence, the induced quantile covariance is also invariant with respect to such transformations in both arguments. 
\begin{proposition}
	\label{properties_quantile_covariance}
For all $\alpha,\beta\in(0,1)$, for all $X,Y\in L^0(\R)$, and for all strictly increasing left-continuous transformations $g\colon\R\to\R$ it holds that
\begin{align*}
\QCov_{\alpha,\beta}\big(g(X),Y\big) &= \QCov_{\alpha,\beta}\big(X,g(Y)\big)  = \QCov_{\alpha,\beta}(X,Y) .
\end{align*}
\end{proposition}

\begin{remark}[Local covariances]
	\label{rem:local_covariances}
Threshold and quantile covariance are closely connected and complementary. 
Indeed, if $(a,b) = (q_{\alpha}(X),q_{\beta}(Y))$, the two measures coincide, $\TCov_{a,b}(X,Y)=\QCov_{\alpha, \beta}(X,Y)$, as the generalised errors coincide, $e_{a}(X)=e_{q_{\alpha}}(X)$ and $e_{b}(Y)=e_{q_{\beta}}(Y)$. 
What distinguishes them is that threshold covariance is a measure on the observation scale, i.e., one chooses a point $(a,b) \in \mathbb{R}^2$, while quantile covariance measures dependence on the quantile scale, i.e., one chooses two quantile levels $(\alpha,\beta) \in (0,1)^2$. If one chooses a point $(a,b)$, computes the respective quantile levels $F_X(a)$ and $F_Y(b)$ to obtain the quantile covariance at this point, this just leads to the threshold covariance at this point, $\QCov_{F_X(a),F_Y(b)} (X,Y) = \TCov_{a,b}(X,Y)$ and vice versa, $\TCov_{q_{\alpha}(X),q_{\beta}(Y)}(X,Y)=\QCov_{\alpha,\beta} (X,Y)$. Both measures allow one to measure dependence \emph{locally} around the specified point: They only depend on the joint exceedances of this point, or in other words, the joint CDF $F_{X,Y}$ or copula $C_{X,Y}$ and the marginal distributions $F_X$, $F_Y$ there. Thus, the corresponding generalised errors are naturally binary random variables, amounting to centred exceedance indicators.
\end{remark}

\begin{example}[Quantile-mean covariance]
\label{exmp:quantile-mean}
We can also pair functionals $\T_1$ and $\T_2$ which belong to different families. E.g., we can define the quantile-mean covariance as 
\be{eq:quantile mean}
\Cov_{q_\alpha,\mu}(X,Y) 
:= \E\big[(\alpha - \one\{X \le q_\alpha(X)\})(Y - \mu(Y))\big] 
= - \E\big[F_{X\,\mid Y}(q_\alpha(X))(Y - \mu(Y))\big],
\ee
where $\alpha\in(0,1)$, $X\in L^0(\R)$ and $Y\in L^1(\R)$.
The second identity is due to the tower property of the conditional expectation, where one first conditions on $Y$.
Conveniently, we may use this definition even if $F_X$ has a jump at its $\alpha$-quantile such that the generalised error for the quantile fails to be centred. The reason is that the generalised error for the mean is always centred.
\end{example}

\begin{remark}\label{rem:last jump 1}
	Naturally, generalised covariances involving the quantile error $e_{q_\alpha}(X)$ \eqref{eq:generalised error quantile} are only interesting 
	for $\alpha$ such that $q_\alpha(X) < \esssup(X) = q_1(X)$. 
	For larger levels $\alpha$, 
	$X$ cannot vary \emph{around} $q_\alpha(X)$, but $X\le q_\alpha(X)$, implying that $e_{q_\alpha}(X)$ is constant and the quantile covariance is 0 for such situations. 
\end{remark}

\section{Generalised correlations}
\label{sec:generalised correlations}

\subsection{Normalisation: Fr\'echet--Hoeffding vs.~Cauchy--Schwarz}
\label{subsec:correlations}

To turn a covariance into a correlation, it is essential to ensure normalisation of the correlation meaning that it only attains values in $[-1,1]$.
Normalisation enhances the interpretability of a correlation. It can be achieved by bounding the covariance in absolute values with a constant $K_{X,Y}$, which may depend on the marginal distributions of $X$ and $Y$
\begin{equation}
\label{eq:bound}
|\Cov_{\T_1,\T_2}(X,Y)| \le K_{X,Y}.
\end{equation}
Then, trivially $\Cov_{\T_1,\T_2}(X,Y)/ K_{X,Y}\in[-1,1]$. 
Pearson correlation $r(X,Y)$ \eqref{eq:Pearson} builds on this idea and exploits the Cauchy--Schwarz inequality to bound the covariance. That is,
for $X,Y\in L^2(\R)$,
\begin{equation}
\label{eq:CSI}
|\Cov(X,Y)| \le K_{X,Y}:= \sqrt{\Var(X) \Var(Y)}.
\end{equation}
For the generalised covariance, we could hence simply exploit the identity \eqref{eq:cov identity} and obtain
\begin{equation}
\label{eq:CSI 2}
|\Cov_{\T_1,\T_2}(X,Y)| \le \sqrt{\E\left[e_{\T_1}(X)^2\right] \E\left[e_{\T_2}(Y)^2\right]},
\end{equation}
provided that $e_{\T_1}(X), e_{\T_2}(Y) \in L^2(\R)$.
The implied quantity is just the Pearson correlation of the generalised errors
\begin{equation}
\label{eq:CSI 3}
r\big(e_{\T_1}(X),e_{\T_2}(Y)\big) = \frac{\E\left[e_{\T_1}(X) e_{\T_2}(Y)\right]}{\sqrt{\E\left[e_{\T_1}(X)^2\right] \E\left[e_{\T_2}(Y)^2\right]}}\in[-1,1].
\end{equation}
However, it is well-known that the Cauchy--Schwarz inequality \eqref{eq:CSI 2} is not  sharp in general. 
Consequently, the Pearson correlation of the generalised errors \eqref{eq:CSI 3} does not always attain all values in $[-1,1]$ for given marginal distributions of $X$ and $Y$. 
In particular, its minimum and maximum may be far away from $-1$ and 1, respectively, and their absolute values may differ strongly, depending on the specific marginal distributions of $X$ and $Y$, see \cite{Embrechts2002}. 
This compromises its interpretability in that it is not really able to indicate the strength of dependence (by the closeness of its absolute value to 1). 
In fact, a value close to 0 may arise even if the dependence is quite strong for certain marginal distributions. 

We revisit the conditions for attainability of Pearson correlation more closely now. 
For $X,Y\in L^2(\R)$, $r(X,Y)$ is $1$ ($-1$) if and only if $X$ and $Y$ have perfect positive (negative) \emph{linear} dependence. 
That is, if and only if there exist $a, a'>0$, $b, b'\in\R$ such that almost surely $Y = aX + b$ ($Y = -a'X  + b'$); see \citet[Theorem 7.28]{mcneil2015}.
Therefore, for given non-degenerate marginals $F_X$ and $F_Y$, there exists a joint distribution $F_{X,Y}$ with Pearson correlation $1$ ($-1$) if and only if $Y$ and $X$ ($-X$) are of the same type, meaning that $Y \stackrel{\mathrm{d}}{=} aX + b$ ($Y \stackrel{\mathrm{d}}{=} -a'X + b'$) for some $a, a'>0$, $b, b'\in\R$. 
This observation leads to the following novel insight: Attainability implies in particular that the marginal distributions are symmetric. 

\begin{lemma}[Attainability of Pearson correlation] 
\label{lem:symmetric}
Let $X,Y\in L^2(\R)$ be non-constant. Pearson correlation $r(X,Y)$ is attainable, that is, there exist joint distributions $F$, $\tilde F$ with marginals $F_X$, $F_Y$ and Pearson correlations $1$ and $-1$, respectively, if and only if $X$ and $Y$ are of the same type \emph{and} the distributions are symmetric, that is, there exist $c,d\in\R$ such that 
$X-c \stackrel{\mathrm{d}}{=} - (X-c)$ and $Y-d \stackrel{\mathrm{d}}{=} - (Y-d)$.
\end{lemma}

The restriction to symmetric distributions of the same type is substantial beyond the confinements of a Gaussian world.
In the context of generalised covariances, which can be written as covariances of generalised errors (see \eqref{eq:cov identity}), it is definitely too restrictive.
Here, the distributions of the generalised errors $e_{\T_1}(X)$ and $e_{\T_2}(Y)$ are generally neither symmetric nor of the same type.

\begin{example} \label{ex:Pearson quantile correlation}
	An example illustrates how severe the attainability problem can become. 
	For the Pearson correlation \eqref{eq:CSI 3} of the quantile errors in the continuous case for $\alpha, \beta \in (0,1)$ from Example \ref{exmp:quantile cov}, which is the measure underlying the quantilogram of \citet{Linton2007}, the upper and lower bound do not even depend on the marginal distributions of $X$ and $Y$, but only on the quantile levels $\alpha$ and $\beta$; see Example \ref{example:Quantile correlation}. Figure \ref{fig:qfcorcsbounds} presents these bounds for all $(\alpha,\beta) \in (0,1)$. 
	An upper (lower) bound of 1 ($-1$) can only be attained for $\alpha=\beta$ ($\alpha=1-\beta$), implying full attainability only for the median case
 $\alpha=\beta=0.5$. Thus, the Cauchy--Schwarz normalisation is usually far too strong for quantile covariance. 
 Moreover, absolute values of the upper and lower bounds can differ greatly. Figure \ref{fig:qfcorcauchyCS} in the Appendix shows what happens to the quantile correlations from the Example in Figure \ref{fig:qfcorcauchy} if one used the Pearson correlation of the quantile errors instead: For all quantile levels $\alpha$ and $\beta$, the measure is very close to 0, impeding an interpretation of the dependence structure.
\end{example}

\begin{figure}
	\centering
	\includegraphics[width=1\linewidth]{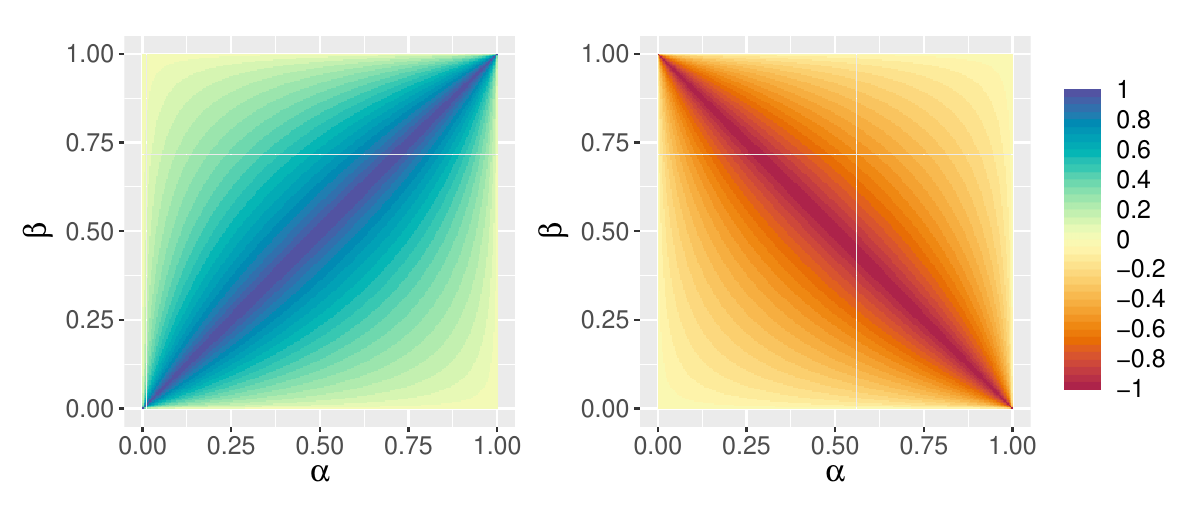}
	\caption[]{Upper and lower bounds for the Pearson correlation of the quantile errors $e_{q_{\alpha}}$ and $e_{q_{\beta}}$}
	\label{fig:qfcorcsbounds}
\end{figure}

The following proposition provides a sharp version of the inequality \eqref{eq:CSI 2},
exploiting previous results by \cite{hoeffding1940},  \cite{Frechet1957}, and \cite{Embrechts2002}.
Recall that $(X,Y)$ are called \emph{comonotonic} if $(X,Y)\dist \big(\nu_1(Z), \nu_2(Z)\big)$ for some random variable $Z$ and two increasing functions $\nu_1$, $\nu_2$.
Similarly, $(X,Y)$ are \emph{countermonotonic} if $(X,Y)\dist \big(\nu_1(Z), \nu_2(Z)\big)$ for some random variable $Z$ with $\nu_1$ increasing and $\nu_2$ decreasing.
In other words, comonotonicity (countermonotonicity) corresponds to the situation of perfect positive (negative) dependence between $X$ and $Y$.

\begin{proposition}
\label{prop:Frechet-Hoeffding}
For any pair of random variables $(X,Y)$, let $(X,Y')$ and $(X,Y'')$ be pairs with the same marginal distributions such that $(X,Y')$ is countermonotonic and $(X,Y'')$ is comonotonic.
Let $\T_1$ and $\T_2$ be functionals with generalised errors $e_{\T_1}$ and $e_{\T_2}$ such that $\Cov_{\T_1,\T_2}(X,Y')$ and $\Cov_{\T_1,\T_2}(X,Y'')$ exist and are finite.
(E.g.\ there are $p,q\in[1,\infty]$ with $p^{-1} + q^{-1}=1$ such that $e_{\T_1}(X) \in L^p(\R)$ and $e_{\T_2}(Y) \in L^q(\R)$.) 
Then the following holds.
\begin{enumerate}[(i)]
\item
$\Cov_{\T_1,\T_2}(X,Y)$ is finite and 
\begin{equation}
\label{eq:Frechet-Hoeffding}
\Cov_{\T_1,\T_2}(X,Y')
\le \Cov_{\T_1,\T_2}(X,Y)
\le \Cov_{\T_1,\T_2}(X,Y'').
\end{equation}
\item
If $e_{\T_1}(X)$ and $e_{\T_2}(Y)$ are non-constant, then 
\begin{equation}
\label{eq:Frechet-Hoeffding 2}
\Cov_{\T_1,\T_2}(X,Y')
<0
< \Cov_{\T_1,\T_2}(X,Y'').
\end{equation}
\item
If $e_{\T_1}(X)$ is a strictly increasing function of $X$ and $e_{\T_2}(Y)$ is a strictly increasing function of $Y$, then an equality in \eqref{eq:Frechet-Hoeffding} is attained only if $(X,Y)$ is co- or countermonotonic.
\end{enumerate}
\end{proposition}

\begin{remark} \label{rem:Perfect local dependence}
Prime examples where the assumption that the generalised errors are \emph{strictly} increasing functions of $X$ and $Y$, respectively, from part (iii) of Proposition \ref{prop:Frechet-Hoeffding} is violated are the local covariances discussed in and before Remark \ref{rem:local_covariances}, where the generalised errors are binary. 
Due to their local nature, they can and should not determine global properties such as co- and countermonotonicity. 
In fact, for them equality in \eqref{eq:Frechet-Hoeffding} holds under perfect \emph{local} dependence, that is, if the corresponding exceedance indicators (or equivalently the generalised errors) are perfectly dependent.
\end{remark}

\begin{example}[Perfect local dependence]
	\label{ex:perfect_local_dependence}
Suppose $X$ follows a uniform distribution on $[0,1]$ and let $Y = 1/2-X$ for $X\in[0,1/2]$ and $Y = 3/2 - X$ for $X\in(1/2,1]$. Then $Y$ is also uniformly distributed on $[0,1]$. 
The pair $(X,Y)$ is neither co- nor countermonotonic.
However, if $\T_1$ and $\T_2$ are the median and if we use the generalised errors induced by \eqref{eq:quantile id}, then $\big(e_{q_{1/2}}(X) , e_{q_{1/2}}(Y)\big)$ only attains the values $(-1/2,-1/2)$ and $(1/2, 1/2)$. 
Hence, the generalised errors are comonotonic, inducing perfect positive dependence locally around the medians such that $\Cov_{q_{1/2},q_{1/2}}(X,Y)
= \Cov_{q_{1/2},q_{1/2}}(X,Y'')$.
\end{example}

\subsection{Definition and properties}

Combining generalised covariance from Definition \ref{defn:generalised covariance} with the Fr\'echet--Hoeffding normalisation from Proposition \ref{prop:Frechet-Hoeffding} leads to generalised correlation.

\begin{definition}[Generalised correlation]
\label{defn:generalised correlation}
Let $T_1\colon\L_1\to\A_1\subseteq\R$, $T_2\colon\L_2\to\A_2\subseteq\R$  be two functionals and $e_{{T_1}}, e_{{T_2}}$ generalised errors for $T_1$ and $T_2$.
Let $\D$ be the set of random variables $(X,Y) \in L^0(\R^2)$ such that 
$e_{{T_1}}(X) e_{{T_2}}(Y')\in L^1(\R)$ and 
$e_{{T_1}}(X) e_{{T_2}}(Y'')\in L^1(\R)$, 
where $(X,Y')$ and $(X,Y'')$ are pairs with the same marginal distributions such that $(X,Y')$ is countermonotonic and $(X,Y'')$ is comonotonic.
If neither $e_{{T_1}}(X)$ nor $e_{{T_2}}(Y)$ are constant almost surely, the \emph{generalised correlation at $\T_1$ and $\T_2$ induced by $e_{{T_1}}$ and $e_{{T_2}}$}, or the \emph{$\T_1-\T_2$-correlation induced by $e_{{T_1}}$ and $e_{{T_2}}$}, 
is defined on $\D$ via
\begin{equation}
\label{eq:generalised cor}
\Cor_{\T_1, \T_2}(X,Y) :=
\begin{cases}
\frac{\Cov_{\T_1, \T_2}(X,Y)}{|\Cov_{\T_1, \T_2}(X,Y'\,)|}, &\text{if } \Cov_{\T_1, \T_2}(X,Y)<0, \\[0.5em]
\frac{\Cov_{\T_1, \T_2}(X,Y)}{|\Cov_{\T_1, \T_2}(X,Y''\,)|}, &\text{if } \Cov_{\T_1, \T_2}(X,Y)\ge0.
\end{cases}
\end{equation}
If one of the generalised errors is constant almost surely, so in particular if $X$ or $Y$ is constant, then
\[
\Cor_{\T_1, \T_2}(X,Y) :=0.
\]
\end{definition}

We summarise the most important properties of generalised correlation.

\begin{theorem}[Properties of generalised correlation]
\label{theorem:properties}
The generalised correlation at $\T_1$ and $\T_2$ induced by $e_{{T_1}}$ and $e_{{T_2}}$ satisfies the following properties.
\begin{enumerate}[(i)]
\item
Normalisation: $\Cor_{\T_1, \T_2}(X,Y) \in [-1,1]$.
\item
Independence implies nullity:
$\Cor_{\T_1, \T_2}(X,Y) =0$ if $X$ and $Y$ are independent.
\item
Perfect dependence:
Suppose $e_{T_1}(X)$ and $e_{T_2}(Y)$ are not constant almost surely.
\begin{enumerate}
\item
$\Cor_{\T_1, \T_2}(X,Y) =1(-1)$ if $X$ and $Y$ are comonotonic (countermonotonic). 
\item
If $e_{T_1}(X)$ is a strictly increasing function of $X$ and $e_{T_2}(Y)$ is a strictly increasing function of $Y$, then $\Cor_{\T_1, \T_2}(X,Y) =1(-1)$ implies that $X$ and $Y$ are comonotonic (countermonotonic).
\end{enumerate}
\item
Symmetry: It holds that 
	$\Cor_{\T_1, \T_2}(X,Y) = \Cor_{\T_2, \T_1}(Y,X)$.	
	In particular, if $T_1 = T_2$,
then $\Cor_{\T_1, \T_2}(X,Y) = \Cor_{\T_1, \T_2}(Y,X)$.
\end{enumerate}
\end{theorem}

Properties (i), (ii) and part (a) of (iii) are fundamental properties that every correlation-type dependence measure should fulfil. 
They ensure interpretability in that the measure takes the right values in the extreme cases of independence and perfect positive and negative dependence. Part (b) of (iii) is not always desirable, e.g., when it comes to local dependence measures, see Remark \ref{rem:Perfect local dependence}.

We would like to highlight the novelty of normalising covariances with Fr\'echet--Hoeffding bounds, which are sharp by construction. 
We are aware of similar constructions only 
in the context of a dependence measure for binary random variables \citep{cole1949} and to combat the non-attainability of Spearman's $\rho$ and Kendall's $\tau$
 in the discrete case 
 \citep{VandenhendeLambert2003, genest2007}.

\subsection{Examples}
\label{subsec:Examples of generalised correlations}

We are going to review the examples of generalised covariances from Subsection \ref{subsec:Examples of generalised covariances} and see how they translate into generalised correlations.
To that end, we will mainly focus on the normalisation terms, that is, the denominators in \eqref{eq:generalised cor}.
Throughout this section, let again $(X,Y)$, $(X,Y')$, $(X,Y'')$ be pairs with the same marginal distribution, where $(X,Y')$ is countermonotonic and $(X,Y'')$ is comonotonic.

\begin{example}[Mean correlation] \label{ex:Mean correlation}
We write $\MCor(X,Y):= \Cor_{\mu,\mu} (X,Y)$ for mean correlation. Due to Hoeffding's formula \cite[Lemma 7.27]{mcneil2015}, it holds that 
\begin{align*}
\Cov(X,Y')&= \iint \max \big(F_X(z_1) + F_Y(z_2) - 1,0\big) - F_X(z_1)F_Y(z_2)\d z_1 \mathrm{d} z_2, \\
\Cov(X,Y'')&= \iint \min \big(F_X(z_1), F_Y(z_2)\big) - F_X(z_1)F_Y(z_2)\d z_1 \mathrm{d} z_2.
\end{align*}
In particular, for the special case that the full range $[-1,1]$ is attainable by Pearson correlation, that is, if $X$ and $Y$ are of the same type and symmetric (Lemma \ref{lem:symmetric}), e.g., under bivariate normality, mean correlation coincides with Pearson correlation, and $|\Cov(X,Y')| = |\Cov(X,Y'')| = \sqrt{\Var(X)\Var(Y)}$, provided that $X,Y\in L^2(\R)$. Thus, indeed Pearson correlation arises as a special case of generalised correlation. Mean correlation, in turn, can be viewed as an improved version of Pearson correlation, which solves the attainability problem.
\end{example}

\begin{example}[Expectile correlation] \label{ex:Expectile_correlation}
	$\ECov_{\tau,\eta}(X,Y')$ and $\ECov_{\tau,\eta}(X,Y'')$ can be calculated via \eqref{eq:cov identity}, i.e., the representation of generalised covariance as covariance of generalised errors, and again Hoeffding's formula.
\end{example}

It follows from Proposition \ref{prop:properties_expectile_cor} that expectile correlation and hence in particular mean correlation is invariant under linear transformations.
\begin{proposition}
	For all $\tau,\eta\in(0,1)$, for all $X,Y\in L^1(\R)$ such that $XY', XY''\in L^1(\R)$, for all $c\in\R$ and $\lambda>0$ it holds that
	\begin{align*}
		\ECor_{\tau,\eta}(\lambda X + c,Y) &= \ECor_{\tau,\eta}(X,\lambda Y + c) = \ECor_{\tau,\eta}(X,Y).
	\end{align*}
\end{proposition}

\begin{example}[Threshold correlation]
\label{example:Threshold correlation}
For the threshold correlation, the classical Fr\'echet--Hoeffding bounds for joint CDFs arise as normalisations, that is,
for $a,b\in\R$
\begin{align*}
\TCov_{a,b}(X,Y')
&= \max\big(F_X(a) + F_Y(b) - 1,0\big) - F_X(a)F_Y(b),\\
\TCov_{a,b}(X,Y'')
&= \min\big(F_X(a), F_Y(b)\big) - F_X(a)F_Y(b).
\end{align*}
\end{example}

\begin{example}[Quantile correlation]
\label{example:Quantile correlation}
For $\alpha,\beta\in(0,1)$ and $X,Y \in L_\text{con}$, 
the Fr\'echet--Hoeffding bounds for copulas arise as normalising terms, that is, 
for $\alpha,\beta\in(0,1)$
\begin{align*}
\QCov_{\alpha,\beta}(X,Y')
= \max (\alpha + \beta - 1,0 ) - \alpha\beta, 
\qquad
\QCov_{\alpha,\beta}(X,Y'')
= \min ( \alpha, \beta ) - \alpha\beta.
\end{align*}
If $F_X$ or $F_Y$ is discontinuous, applying the Fr\'echet--Hoeffding bounds to \eqref{eq:Quantile correction copula} yields
\begin{align*}
\QCov_{\alpha,\beta}(X,Y')
&= \max \big(F_X(q_\alpha(X)) + F_Y(q_\beta(Y)) - 1,0 \big) - F_X(q_\alpha(X))F_Y(q_\beta(Y)), 
\\
\QCov_{\alpha,\beta}(X,Y'')
&= \min \big( F_X(q_\alpha(X)),F_Y(q_\beta(Y)) \big) - F_X(q_\alpha(X))F_Y(q_\beta(Y)).
\end{align*}
\end{example}

The following result follows from Proposition \ref{properties_quantile_covariance_continuous}.

\begin{proposition}
	\label{properties_quantile_correlation_continuous}
	For all $\alpha,\beta\in(0,1)$, for all $X,Y\in L_{\text{con}}$ with $F_X$ and $F_Y$ strictly increasing on their support, it holds that
	\begin{align*}
		\QCor_{\alpha,\beta}(X,Y) &= -\QCor_{\alpha,1-\beta}(X,-Y) =  -\QCor_{1-\alpha,\beta}(-X,Y).
	\end{align*}
\end{proposition}

Quantile correlation is invariant with respect to strictly increasing transformations, which follows from Proposition \ref{properties_quantile_covariance}. 
Thus, quantile correlation belongs to the family of rank correlations.
\begin{proposition}
\label{prop:rank correlations}
	For all $\alpha,\beta\in(0,1)$, for all $X,Y\in L^0(\R)$, and for all strictly increasing transformations $g\colon\R\to\R$ it holds that
	\begin{align*}
	\QCor_{\alpha,\beta}\big(g(X),Y\big) &= \QCor_{\alpha,\beta}\big(X,g(Y)\big)  = \QCor_{\alpha,\beta}(X,Y) .
	\end{align*}
\end{proposition}
\begin{example}[Median correlation]
	Median correlation, $\QCor_{0.5,0.5}$, arises as a special case of quantile correlation. 
	What is particular about it is that, for continuous marginals $F_X$ and $F_Y$, 
	 the generalised errors $e_{q_{0.5}}(X)$ and $e_{q_{0.5}}(Y)$ are symmetric and of the same type. 
	 Hence, they fulfil the conditions of Lemma \ref{lem:symmetric}, implying that the Cauchy--Schwarz and the Fr\'echet--Hoeffding normalisation coincide, 
	 leading to
	 $\QCov_{\alpha,\beta}(X,Y')=-\QCov_{\alpha,\beta}(X,Y'')=\sqrt{\Var(e_{q_{0.5}}(X))\Var(e_{q_{0.5}}(Y))}=1/4$. Indeed, then median correlation is equal to \citeauthor{Blomqvist1950}'s (\citeyear{Blomqvist1950}) $\beta$, 
defined as	 $\beta(X,Y) := 4 \P (X \leq q_{0.5}(X), Y \leq q_{0.5}(Y)) - 1$, which is also sometimes referred to as median correlation in the literature.
\end{example}

Quantile correlation does not only generalise median correlation. When moving from the center to the tails, that is, considering the limit of $\QCor_{\alpha,\alpha}$ for $\alpha \to 0$ or $\alpha \to 1$, a quantity closely related to the well-known coefficient of tail dependence shows up which is discussed in Section \ref{sec:tail}.

\begin{example}[Quantile-mean correlation]
Again, $\Cov_{q_\alpha,\mu}(X,Y')$ and $\Cov_{q_\alpha,\mu}(X,Y'')$ can be calculated via \eqref{eq:cov identity} and Hoeffding's formula. If $X$ and $Y$ have continuous and strictly increasing marginal distributions, they take a particularly convenient form:
\begin{align*}
\Cov_{q_\alpha,\mu}(X,Y') = \E\big[(\alpha - \one\{X \le q_\alpha(X)\})(Y' - \mu(Y'))\big] = \alpha\big(\mu(Y) - \ES_{1-\alpha}^+(Y)\big)<0,
\end{align*}
where $\ES_{1-\alpha}^+(Y)$ is the upper expected shortfall $\ES_{1-\alpha}^+(Y) := \E[Y|Y> q_{1-\alpha}(Y)]$, which equals $\frac{1}{\alpha}\E[Y \one\{Y > q_{1-\alpha}(Y)\}]$,
and
\begin{align*}
\Cov_{q_\alpha,\mu}(X,Y'') = \E\big[(\alpha - \one\{X \le q_\alpha(X)\})(Y'' - \mu(Y''))\big] = \alpha\big(\mu(Y) - \ES_\alpha^-(Y)\big) >0,
\end{align*}
where $\ES_\alpha^-(Y)$ is the lower expected shortfall $\ES_\alpha^-(Y) := \E[Y|Y\le q_\alpha(Y)]$, which equals $\frac{1}{\alpha}\E[Y \one\{Y \le q_\alpha(Y)\}]$.
\end{example}

\section{Distributional covariances and correlations}
\label{sec:local_distributional}

\subsection{Definition and properties}

The generalised covariances and correlations presented so far  measure the dependence of two random variables $X$, $Y$ around the two functionals $T_1(X)$, $T_2(Y)$. 
As such, they focus on a certain aspect of the dependence structure of $X$ and $Y$. 
This section proposes measures that uncover the \emph{entire} dependence structure of $X$ and $Y$. 
The idea behind those measures is to consider all the local information contained in the whole families of local covariances (correlations) jointly (see Remark \ref{rem:local_covariances}), leading to dependence measures that are functions in two arguments -- the thresholds $(a,b) \in \mathbb{R}^2$ or quantile levels $(\alpha,\beta) \in (0,1)^2$. 

An alternative way to arrive at the same measures is to approach them from the angle of generalised covariances and correlations by not considering point-valued functionals $T_1(X)$ and $T_2(Y)$ as in Sections \ref{sec:generalised} and \ref{sec:generalised correlations}, but the entire CDF or the quantile function themselves.
An $L^0(\R)$-identification function for the identity or CDF-functional is the function-valued map 
$v_{\CDF}(F,x) = \big(F(a) - \one\{x\le a\}\big)_{a\in\R}$, leading to the function-valued generalised error $e_{\CDF}(X) = \big(F(a) - \one\{X\le a\}\big)_{a\in\R}$. 
Alternatively, one may consider the quantile function (QF) functional, mapping a CDF, $F$, to its generalised inverse, $F^{-1}$. 
On the class of random variables with a continuous CDF, denoted by $L_{\text{con}}$, we have the $L_{\text{con}}$-identification function
$v_{\QF}(F^{-1},x) = \big(\alpha - \one \{x \le F^{-1}(\alpha)\}\big)_{\alpha\in(0,1)}$. 
The modification discussed in Example \ref{example:discontinuous marginals} for the general case leads to the generalised error $e_{\QF}(X) := \big( F_X( q_\alpha (X) ) - \one\{X \le q_\alpha(X)\} \big)_{\alpha\in(0,1)}$.
We can construct generalised covariances (and correlations) from $e_{\CDF}$ and $e_{\QF}$ via an outer product ansatz, generalising \eqref{eq:generalised cov}. 
This reasoning explains why we call the resulting function-valued dependence measures CDF and quantile function covariance (correlation), and subsume both under the name distributional covariances (correlations).

\begin{definition}[Distributional covariances and correlations]
For any $X,Y\in L^0(\R)$, the \emph{CDF covariance} and the \emph{CDF correlation} are 
\begin{align*}
&\CDFCov(X,Y)\colon \R^2 \to \R, &&(a,b) \mapsto \TCov_{a,b}(X,Y),\\
&\CDFCor(X,Y)\colon \R^2 \to [-1,1], &&(a,b) \mapsto \TCor_{a,b}(X,Y).
\end{align*}
Likewise, the \emph{quantile function covariance} and the \emph{quantile function correlation} are 
\begin{align*}
&\QFCov(X,Y)\colon (0,1)^2 \to \R, &&(\alpha, \beta) \mapsto \QCov_{\alpha, \beta}(X,Y),\\
&\QFCor(X,Y)\colon (0,1)^2 \to [-1,1], &&(\alpha, \beta) \mapsto \QCor_{\alpha, \beta}(X,Y).
\end{align*}
\end{definition}  

Just as in univariate statistics when characterising the distribution of a single random variable, we may either take the perspective of the quantile function or of the CDF when characterising dependence between two random variables, as there are two corresponding families of distributional dependence measures. 
Indeed, the distributional covariances and correlations characterise the dependence structure between $X$ and $Y$ fully: Given the marginals $F_X$ and $F_Y$, it is possible to recover the joint CDF $F_{X,Y}$ or the copula $C_{X,Y}$, which contain the information on the full dependence structure, from $\CDFCor$ or $\QFCor$ (as well as $\CDFCov$ or $\QFCov$). 
This is clear from the representations of $\TCov$ in terms of $F_{X,Y}$ and the marginals in \eqref{eq:Threshold cov} and $\QCov$ in terms of the $C_{X,Y}$  and the marginals in \eqref{eq:Quantile correction copula} as well as the normalisations, which only depend on the marginals, see Examples \ref{example:Threshold correlation} and \ref{example:Quantile correlation}. If $F_X$ and $F_Y$ are continuous, $\QFCor$ and $\QFCov$ are independent of the marginals, and by \eqref{eq:quantile copula} and again the normalisations from Example \ref{example:Quantile correlation}, there is even a one-to-one mapping between $\QFCov$ and $C_{X,Y}$ (as well as $\QFCor$ and $C_{X,Y}$). 
While of course the joint CDF and copula contain the information on the dependence structure as well, we emphasise that particularly the normalised quantities, that is, the distributional correlations, have the great advantage that they are really dependence measures. 
Thus, they uncover the full dependence structure, and regions of stronger and weaker dependence can be identified relatively easily, see Subsection \ref{subsec:examples distr corr} and Section \ref{subsec:data examples} for theoretical and empirical examples, respectively. 

Distributional correlations are able to characterise the limiting cases of independence and perfect positive and negative dependence properly. In particular, nullity of any of the distributional covariances or correlations implies independence (Proposition \ref{prop:nullity}). Further, unity (negative unity) of distributional correlations implies perfect positive (negative) dependence (Corollary \ref{cor:distributional properties}).
Additionally, $\QFCor$ is invariant to strictly increasing transformations.

\begin{proposition}\label{prop:nullity}
	For any $X,Y\in L^0(\R)$ and for any of the four distributional dependence measures $\mathrm{D} \in \{\CDFCov, \CDFCor, \QFCov, \QFCor\}$ it holds that $X$ and $Y$ are independent if and only if $\mathrm D(X,Y)$ vanishes identically. 
\end{proposition}

\begin{remark}\label{rem:last jump 2}
	The functions $\CDFCov, \CDFCor, \QFCov$ and $\QFCor$ are naturally constant between jumps of the joint CDF $F_{X,Y}$, see \eqref{eq:Threshold cov} and \eqref{eq:Quantile correction copula}. 
	Furthermore, along the lines of Remark \ref{rem:last jump 1}, the generalised errors are constant for values $a\ge \esssup(X)$ and $\alpha$ such that $q_\alpha(X)=q_1(X) = \esssup(X)$. 
	Thus, $\CDFCov$ and $\CDFCor$ are naturally only interesting on the restricted domain $[\essinf(X),\esssup(X)) \times [\essinf(Y),\esssup(Y))$ and $\QFCov$ and $\QFCor$ on $[F_X(\essinf(X)), F_X(\esssup(X))) \times [F_Y(\essinf(Y)), F_Y(\esssup(Y)))$.
\end{remark}

\begin{corollary}\label{cor:distributional properties}
	For any $X,Y \in L^0(\mathbb{R})$ and any of the two distributional correlations $\DCor \in \{\QFCor,\CDFCor\}$ it holds that
	\begin{enumerate}[(i)]
		\item
		Normalisation: $\DCor(X,Y) \in [-1,1]$.
		\item
		Perfect dependence:
		$\CDFCor(X,Y) =1(-1)$ on $[\essinf(X),\esssup(X)) \times [\essinf(Y),\esssup(Y))$ if and only if $X$ and $Y$ are comonotonic (countermonotonic). \\
		$\QFCor(X,Y) =1(-1)$ on $[F_X(\essinf(X)), F_X(\esssup(X))) \times [F_Y(\essinf(Y)), F_Y(\esssup(Y)))$ if and only if $X$ and $Y$ are comonotonic (countermonotonic).
	\end{enumerate}
\end{corollary}

\subsection{Examples}
\label{subsec:examples distr corr}

We already discussed an example of quantile function correlation for a Cauchy copula depicted in Figure \ref{fig:qfcorcauchy} in the introduction. It showed that despite a Spearman's $\rho$ of 0, there is quite an interesting dependence structure, which quantile function correlation uncovers.  

\begin{figure}[ht]
	\centering
	\includegraphics[width=1\linewidth]{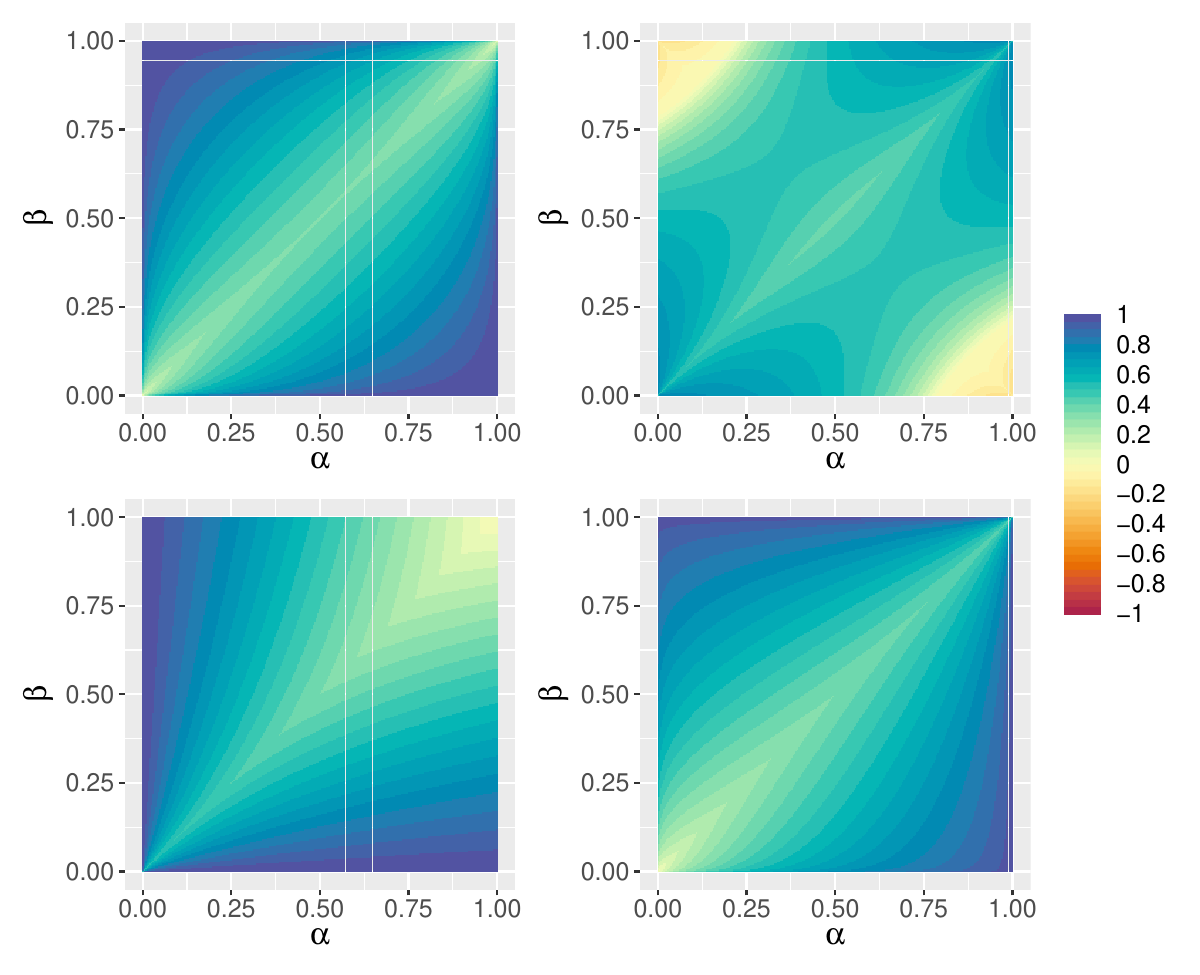}
	\caption{Plots of $\QFCor$ for four different copulas, all with Spearman's $\rho$ equal to 0.5: Gaussian (upper left), Cauchy (upper right), Clayton (lower left) and Gumbel (lower right)}
	\label{fig:qfcorvarious}
\end{figure}

Figure \ref{fig:qfcorvarious} presents quantile function correlations stemming from four further classical copulas, all having a Spearman's $\rho$ of 0.5: a Gaussian, a Cauchy, a Clayton, and a Gumbel copula. The corresponding scatter plots with simulated samples of size 1000 can be found in Figure \ref{fig:qfcorvariousscatter} in the Appendix.
Even though they share the same Spearman's $\rho$,
the copulas exhibit very different quantile function correlations. 
For example, the Gaussian copula has a weak dependence in the lower left and upper right corner, while the Cauchy copula shows a strong dependence in both corners, and the Clayton and Gumbel copula exhibit strong dependence in one corner, but not the other. 
In the upper left and lower right corners, for all but the Cauchy copula there is strong positive dependence, which means, e.g., for the lower right corner that exceedances of a high quantile of $X$ are strongly positively associated with exceedances of a small quantile of $Y$. 
This means
that very large values of $X$ and very small values of $Y$ virtually never occur jointly for those copulas. 
In contrast, for the Cauchy copula, despite the overall positive dependence between $X$ and $Y$ as indicated by Spearman's $\rho$, the dependence in the lower right and upper left corners even becomes negative, reflecting the tail behaviour of the Cauchy distribution also seen in Figure \ref{fig:qfcorcauchy}. 

Figure \ref{fig:qfcorvariousnegative} in the Appendix illustrates the symmetry of quantile function correlation of Proposition \ref{properties_quantile_correlation_continuous}, in the cases of the Gaussian and Cauchy copula with Spearman's $\rho$ of 0.5 and $-0.5$.

Using an alternative version of quantile function correlation employing the Cauchy--Schwarz normalisation leads to a misrepresentation of the dependence structure as discussed in Example \ref{ex:Pearson quantile correlation}. Figure \ref{fig:qfcorvariousCS} in the Appendix plots this function for the four copulas discussed above. Comparing it to Figure \ref{fig:qfcorvarious} illustrates this point further.

Note that quantile function and the closely related CDF correlation (see Remark \ref{rem:local_covariances}) coincide for copulas due to the uniformity of the marginals.
For other marginals, they usually look like distorted versions of each other, see Section \ref{subsec:data examples} for examples and further discussion.

\subsection{Global positive and negative dependence}

There is a vast literature on the question what it means or what it should mean that two random variables $X$ and $Y$ are \emph{globally} positively (negatively) dependent. 
We refer to \cite{Mari2001} and \cite{Balakrishnan2009} for overviews of this strand of literature. The distributional covariances and correlations suggest a natural definition. 

\begin{definition}
\label{definition:dependence}
Any $X,Y\in L^0(\R)$ are \emph{globally positively dependent} if $\CDFCov(X,Y)\ge0$. They are \emph{globally negatively dependent} if $\CDFCov(X,Y)\le0$.
\end{definition}

\begin{remark}
Definition \ref{definition:dependence} of global positive  dependence coincides with \citeauthor{lehmann1966}'s (\citeyear{lehmann1966}) definition of positive  \emph{quadrant dependence}, which holds for two random variables $X,Y$ if 
\begin{equation*}
	\p (X \leq a, Y \leq b) \geq \p (X \leq a) \p (Y \leq b) \quad \text{ for all } a,b \in \mathbb{R}.
\end{equation*}
\end{remark}
The following proposition shows that we can define global positive and negative dependence also in terms of quantile function correlation.

\begin{proposition} \label{prop:global dependence}
	Any $X,Y\in L^0(\R)$ are globally positively (negatively) dependent if and only if $\QCov(X,Y)\ge0$ ($\QCov(X,Y)\le0$).
\end{proposition}

For example, the first three copulas from Figure \ref{fig:qfcorvarious} are globally positively dependent, while the Cauchy copulas in Figures \ref{fig:qfcorcauchy} and \ref{fig:qfcorvarious} represent cases of mixed dependence, with regions of local positive as well as local negative dependence.

Each generalised covariance and correlation directly gives rise to a specific concept of positive (negative) dependence between two random variables $(X,Y)$ as well. 
The following proposition shows that such a dependence is implied by global dependence.

\begin{proposition}
\label{prop:dependence}
	Assume that $X,Y\in L^0(\R)$ and $T_1,T_2$ are such that the generalised covariance $\Cov_{T_1,T_2}(X,Y)$ exists. If $X$ and $Y$ are globally positively dependent, it holds that 
	$\Cov_{T_1,T_2}(X,Y) \geq 0.$
	If $X$ and $Y$ are globally negatively dependent, 
	$\Cov_{T_1,T_2}(X,Y) \leq 0$.
\end{proposition}

Proposition \ref{prop:dependence} can also be viewed as a complement to Theorem \ref{theorem:properties} in that it contains a further desirable property of generalised correlations: 
They not only take the correct values of $-1$\,/\,0\,/\,1 in the extreme cases, but also the correct sign under positive and negative dependence.

\section{Tail correlations and tail dependence}
\label{sec:tail}

It is often of interest to analyse co-movement in the tails, i.e., if for example very large values in $X$ and $Y$ tend to occur together, if there is no dependence in the tails, or if large values of $X$ render large values of $Y$ rather more unlikely and vice versa. 
Such questions are subsumed under the term of tail, extremal or asymptotic dependence in the literature, see, e.g., \citet[Chapter 2]{Joe2014} or \cite{Coles1999}. 

Natural measures for lower and upper tail dependence are the respective limits of quantile correlation or the limits of quantile function correlation when moving to the upper right and lower left corner, respectively.
 
\begin{definition}[Tail correlations]
	\label{definition:tailcorr}
	For any $X,Y \in L^0(\mathbb{R})$ the lower and upper tail correlations are defined as 
	\begin{equation*}
	\LTCor(X,Y) := \lim_{\alpha \downarrow 0} \QCor_{\alpha,\alpha}(X,Y), \qquad 
	\UTCor(X,Y) := \lim_{\alpha \uparrow 1} \QCor_{\alpha,\alpha}(X,Y),
	\end{equation*}
	respectively, provided that the limits exist.
\end{definition}
On top of the upper and lower tail correlation, one might also consider 
$\lim_{\alpha \downarrow 0} \QCor_{\alpha,1-\alpha}(X,Y)$ and 
$\lim_{\alpha \downarrow 0} \QCor_{1-\alpha,\alpha}(X,Y)$, provided that these limits exist. 
The discussion is similar to what follows and is therefore omitted.

\begin{definition}[Tail dependence]
	\label{definition:taildependence}
	Any $X,Y \in L^0(\mathbb{R})$ are \emph{positively lower tail dependent} if $\LTCor(X,Y) > 0$, \emph{negatively lower tail dependent} if $\LTCor(X,Y) < 0$, and \emph{lower tail independent} if $\LTCor(X,Y) = 0$. 
	They are \emph{lower tail comonotonic} if $\LTCor(X,Y) = 1$ and \emph{lower tail countermonotonic} if $\LTCor(X,Y) = -1$. 
	For the upper tail notions, replace $\LTCor(X,Y)$ by $\UTCor(X,Y)$.  
\end{definition}

The by far most prominent measure of tail dependence is the coefficient of tail dependence, see \cite{Joe1993}, \cite{Coles1999}. \cite{Fiebig2017} 
contains a literature review on the use and naming of the coefficient in different fields. To facilitate the following discussion of the relation between coefficients of tail dependence and tail correlations we assume continuity of the marginals $F_X$ and $F_Y$ throughout, stated by $X, Y\in L_{\mathrm{con}}$. 
The coefficients of lower and upper tail dependence are defined as
\[
\lambda_l(X,Y) := \lim_{\alpha \downarrow 0} \p ( Y \leq q_{\alpha} (Y) \mid X \leq q_{\alpha} (X)) = \lim_{\alpha \downarrow 0} \frac{C_{X,Y}(\alpha,\alpha)}{\alpha}, \quad X, Y\in L_{\text{con}},
\]
and
\[
\lambda_u(X,Y) := \lim_{\alpha \uparrow 1} \p ( Y > q_{\alpha} (Y) \mid X > q_{\alpha} (X)) = \lim_{\alpha \uparrow 1} \frac{\overline{C}_{X,Y}(\alpha,\alpha)}{ 1 - \alpha}, \quad X, Y\in L_{\text{con}},
\]
where $\overline{C}_{X,Y}$ denotes the survival function of the copula $C_{X,Y}$
and where we assume that the limits exist.
The following lemma clarifies the relation between the coefficients of tail dependence and the tail correlations.

\begin{lemma}\label{lem:tail dependence}
For $X, Y\in L_{\mathrm{con}}$, the following assertions hold.
\begin{enumerate}[(a)]
\item
If the coefficient of lower (upper) tail dependence or the lower (upper) tail correlation exist and are positive, the other quantity exists as well and the two quantities coincide.
That is,
\begin{align*}\nonumber
\lambda_l(X,Y)
&= \lim_{\alpha \downarrow0} \frac{C_{X,Y}(\alpha,\alpha)}{\alpha}
= 
 \lim_{\alpha \downarrow0} \frac{C_{X,Y}(\alpha,\alpha) - \alpha^2}{\alpha - \alpha^2} 
 =\LTCor(X,Y)\,,\\
 \lambda_u(X,Y)
&= \lim_{\alpha \uparrow 1} \frac{\overline{C}_{X,Y}(\alpha,\alpha)}{1-\alpha} 
= \lim_{\alpha \uparrow 1} \frac{\overline{C}_{X,Y}(\alpha,\alpha) - (1-\alpha)^2}{\alpha(1-\alpha)} 
=  \lim_{\alpha \uparrow 1} \frac{C_{X,Y}(\alpha,\alpha) - \alpha^2}{\alpha - \alpha^2}\\
&= \UTCor(X,Y)\,.
\end{align*}
\item
If $X, Y$ are lower (upper) tail independent, the coefficient of lower (upper) tail dependence is 0 as well.
\item
If $X,Y$ are negatively lower (upper) tail dependent, it holds that 
\begin{align*}
\LTCor(X,Y)  &= \lim_{\alpha \downarrow 0} \frac{C_{X,Y}(\alpha,\alpha) - \alpha^2}{\alpha^2} = \lim_{\alpha \downarrow 0} \frac{C_{X,Y}(\alpha,\alpha)}{\alpha^2}-1 \text{ and } \lambda_l(X,Y) = 0 \,,\\
\UTCor(X,Y) &= \lim_{\alpha \uparrow 1} \frac{C_{X,Y}(\alpha,\alpha) - \alpha^2}{(1-\alpha)^2} 
= \lim_{\alpha \uparrow 1} \frac{\overline{C}_{X,Y}(\alpha,\alpha)}{(1-\alpha)^2} - 1 \text{ and } \lambda_u(X,Y) = 0 \,.
\end{align*}
\end{enumerate}
\end{lemma}

In the literature, the cases $\lambda_l=0$ and $\lambda_u=0$ are usually called \textit{tail} or \textit{asymptotic independence} and $\lambda_l>0$ and $\lambda_u>0$ tail dependence \citep{mcneil2015}.
The lower and upper tail correlations provide a more nuanced picture of the tail behaviour. 
While under positive tail dependence as introduced in Definition \ref{definition:taildependence}, the coefficients of tail dependence and the tail correlations coincide, the latter measures are able to classify the situation of $\lambda_l=0$ and $\lambda_u=0$ into actual tail independence (the tail correlations are 0) on the one hand and \emph{negative} tail dependence on the other hand -- also indicating the strength of negative dependence. Under positive tail dependence, a tail event (e.g., very high or very low values) in $X$ makes a tail event in $Y$ more likely and vice versa, while under negative tail dependence, it makes the other one less likely. Under tail independence, the tail events do not influence each other.

In the example from Figure \ref{fig:qfcorvarious}, the upper and lower tail correlations are both 0 for the Gaussian, both $0.536$ for the Cauchy, $0$ and $0.525$ for the Clayton, and $0.433$ and $0$ for the Gumbel copula. They coincide with the values for the corresponding coefficients of tail dependence as they are non-negative. 

\begin{example}[Tail independence vs.\ negative tail dependence]
Consider a pair of bivariate normal random variables $(X,Y)$ with $0<\Cor(X,Y)<1$ and  the pair $(X,-Y)$, which is also bivariate normal with negative correlation $\Cor(X,-Y)=-\Cor(X,Y)$, see the top panels in Figure \ref{fig:qfcorvariousnegative} in the Appendix for a visualization of the corresponding quantile function correlations. In both cases it holds that $\lambda_l=\lambda_u=0$. However, the two copulas show a fundamentally different tail behaviour: For the pair with positive Pearson correlation, $(X,Y)$, the points in the upper right and lower left corner scatter unsystematically (see the upper left panel of Figure \ref{fig:qfcorvariousscatter} in the Appendix), which is due to tail independence as indicated by $\UTCor(X,Y)=\LTCor(X,Y)=0$. For the pair with negative correlation, $(X,-Y)$, however, an extremely high (low) value in $X$ precludes an extremely high (low) value in $Y$.
That means in a scatter plot of a bivariate normal with negative correlation, no points lie in the upper right and lower left corners, which is due to tail countermonotonicity indicated by $\UTCor(X,-Y)=\LTCor(X,-Y)=-1$, the strongest form of negative tail dependence.

In the extreme case $\Cor(X,Y)=1$, $X$ and $Y$ are comonotone and thus $\lambda_l(X,Y)=\lambda_u(X,Y)=\UTCor(X,Y)=\LTCor(X,Y)=1$, while $X$ and $-Y$ are countermonotone and as above $\lambda_l(X,-Y)=\lambda_u(X,-Y)=0$ and $\UTCor(X,-Y)=\LTCor(X,-Y)=-1$.
\end{example}

\section{Summary covariances and correlations}
\label{sec:summary}

\subsection{Summary covariances}

The distributional covariances and correlations from Section \ref{sec:local_distributional} reveal the full dependence structure between $X$ and $Y$. 
Nevertheless, it is often required or useful to summarise the dependence structure in a single number. 
In fact, this is what most classical dependence measures aim for. 
A natural way to construct such summary measures from distributional covariances is to compute weighted averages.

\begin{definition}[Summary covariances]
	\label{definition:summarycovariances}
For any $X,Y\in L^0(\R)$ the summary covariance induced by quantile function covariance with respect to a measure $\kappa$ on $[0,1]^2$
is
\begin{equation}
\label{eq:summary QF}
\SCov_{\QF,\kappa} (X,Y) := \int_{[0,1]^2}  \QCov_{\alpha,\beta} (X,Y) \d\kappa (\alpha,\beta).
\end{equation}
Likewise, the summary covariance induced by CDF covariance with respect to a measure $\nu$ on $\R^2$ is
\begin{equation}
\label{eq:summary CDF}
\SCov_{\CDF,\nu} (X,Y) := \int_{\mathbb{R}^2}  \TCov_{a,b} (X,Y) \d\nu (a,b).
\end{equation}
\end{definition}
We tacitly assume that the integrals in \eqref{eq:summary QF} and \eqref{eq:summary CDF} exist and are finite. 
Since both $\QCov$ and $\TCov$ are bounded, a sufficient condition is that $\nu$ and $\kappa$ are finite.
The summary covariances inherit the properties of the distributional covariances. They are 0 under independence and nonnegative (nonpositive) under global positive (negative) dependence. Further, the quantile function summary covariance is invariant under strictly increasing transformations.

Interestingly, two of the most popular dependence measures arise as canonical special cases of summary covariances.

\begin{example}[Covariance]
	\label{example:Pearson_Covariance}
	If we plug in the Lebesgue measure $\lambda$ for $\nu$ in \eqref{eq:summary CDF}, Hoeffding's formula \cite[Lemma 7.27]{mcneil2015} implies that 
	\begin{equation*}
	\SCov_{\CDF,\lambda} (X,Y) = \int_{\mathbb{R}^2}  \TCov_{a,b} (X,Y) \d (a,b) = \Cov(X,Y),
	\end{equation*}
	which is the classical covariance.
\end{example}

\begin{example}[Spearman covariance]
	\label{example:Spearman_Covariance}
	Recall that for continuous $F_X$ and $F_Y$, Spearman's rank correlation coefficient $\rho$ can be defined as the Pearson correlation of the probability integral transforms,
	\begin{equation} \label{eq:Spearman}
		\rho(X,Y) = \Cor \big(F_X(X),F_Y(Y)\big). 
	\end{equation}
	From Example \ref{example:Pearson_Covariance} and the relation between quantile and threshold correlation discussed in Remark \ref{rem:local_covariances}, it follows for $\kappa$ being the Lebesgue measure $\lambda$ that \eqref{eq:summary QF} becomes
	\begin{equation*}
	\SCov_{\QF,\lambda} (X,Y) = \int_{[0,1]^2}  \QCov_{\alpha,\beta} (X,Y) \d(\alpha, \beta) = \Cov\big(F_X(X),F_Y(Y)\big),
	\end{equation*}
	which is the Spearman covariance.
\end{example}

We discuss further examples, which focus on specific regions of interest, when dealing with the respective correlations below.

\subsection{Summary correlations}

Summary correlations arise when summary covariances are appropriately normalised, again utilising co- and countermonotonic random couplings with identical marginals as $X$ and $Y$.	

\begin{definition}[Summary correlations]
	\label{defn:summary_correlations}
	Consider measures $\kappa$ on $[0,1]^2$ and $\nu$ on $\mathbb{R}^2$. For any non-constant $X,Y\in L^0(\R)$ consider pairs with the same marginal distributions $(X,Y')$ and $(X,Y'')$ such that $(X,Y')$ is countermonotonic and $(X,Y'')$ is comonotonic. Then the summary correlation induced by quantile function covariance with respect to $\kappa$
	is
	\begin{equation}\label{eq:SCor a}
	\SCor_{\QF,\kappa} (X,Y) := 
	\begin{cases}
	\frac{\SCov_{\QF,\kappa} (X,Y)}{|\SCov_{\QF,\kappa} (X,Y'\,)|}  &\text{if } \SCov_{\QF,\kappa} (X,Y) <0, \\[0.5em]
	\frac{\SCov_{\QF,\kappa} (X,Y)}{|\SCov_{\QF,\kappa} (X,Y''\,)|}  &\text{if } \SCov_{\QF,\kappa} (X,Y)\ge0.
	\end{cases}
	\end{equation}
	Likewise, the summary correlation induced by CDF covariance with respect to $\nu$ is 
	\begin{equation}\label{eq:SCor b}
	\SCor_{\CDF,\nu} (X,Y) := 
	\begin{cases}
	\frac{\SCov_{\CDF,\nu} (X,Y)}{|\SCov_{\CDF,\nu} (X,Y'\,)|}  &\text{if } \SCov_{\CDF,\nu} (X,Y) <0, \\[0.5em]
	\frac{\SCov_{\CDF,\nu} (X,Y)}{|\SCov_{\CDF,\nu} (X,Y''\,)|}  &\text{if } \SCov_{\CDF,\nu} (X,Y)\ge0.
	\end{cases}
	\end{equation}
Provided that the involved quantities exist and are finite. 
If $X$ or $Y$ is constant, then the two measures are set to be 0.
\end{definition}

The normalisation terms in \eqref{eq:SCor a} and \eqref{eq:SCor b} can be computed from \eqref{eq:summary QF} and \eqref{eq:summary CDF} and the Fr\'echet--Hoeffding bounds presented in Examples \ref{example:Quantile correlation} and \ref{example:Threshold correlation}.
The summary correlations inherit the appealing properties of the distributional correlations.
\begin{corollary}
\label{cor:summary correlations}
		For any $X,Y \in L^0(\mathbb{R})$ and for any of the two summary correlations $\SCor \in \{\SCor_{\QF,\kappa},\SCor_{\CDF,\nu}\}$ such that $\SCor(X,Y)$ exists, the following properties hold.
	\begin{enumerate}[(i)]
		\item
		Normalisation: $\SCor(X,Y) \in [-1,1]$.
		\item
		Independence implies nullity:
		$\SCor(X,Y) =0$ if $X$ and $Y$ are independent.
		\item
		Perfect dependence:
		$\SCor(X,Y) =1(-1)$ if $X$ and $Y$ are comonotonic (countermonotonic). 
		If $\kappa$ and $\nu$ are strictly positive (i.e., they assign a strictly positive mass to any open non-empty set),
		then $\SCor(X,Y) =1(-1)$ implies that $X$ and $Y$ are comonotonic (countermonotonic).
		\item
		Symmetry: 
		If $\kappa$ and $\nu$ are invariant in their arguments in the sense that $\d\kappa (\alpha,\beta) = \d\kappa (\beta,\alpha)$ and $\d\nu (a,b) = \d\nu (b,a)$ for all $\alpha, \beta\in [0,1]$ and for all $a,b\in\R$, 
		then $\SCor(X,Y) = \SCor(Y,X)$.
	\end{enumerate}
\end{corollary}

By Proposition \ref{properties_quantile_covariance}, $\SCor_{\QF,\kappa}$ is invariant under strictly increasing transformations of $X$ and $Y$ as well.

\begin{example}[Mean and Pearson correlation]
	\label{example:Summary_mean_correlation}
	From Example \ref{example:Pearson_Covariance} it follows that 
\begin{equation*}
\SCor_{\CDF,\lambda} (X,Y) = \MCor(X,Y),
\end{equation*}
where $\lambda$ is the Lebesgue measure on $\R^2$.
Under the conditions of Lemma \ref{lem:symmetric} it holds that $\SCor_{\CDF,\lambda} (X,Y) = \Cor(X,Y)$.
\end{example}

\begin{example}[Spearman correlation]
	\label{example:Spearman_correlation}
By Example \ref{example:Spearman_Covariance} it holds for the Lebesgue measure $\lambda$ on $[0,1]^2$ that 
\begin{equation*}
\SCor_{\QF,\lambda} (X,Y) = \MCor\big(F_X(X),F_Y(Y)\big).
\end{equation*}
If $F_X$ and $F_Y$ are continuous, Spearman's $\rho$ arises, $\SCor_{\QF,\lambda} (X,Y)=\rho(X,Y)$, since the probability integral transforms are standard uniform, $F_X(X) \sim U(0,1)$, $F_Y(Y) \sim U(0,1)$, and thus fulfil the conditions of Lemma \ref{lem:symmetric}. In this case the normalisation does not depend on the sign of $\Cov(F_X(X),F_Y(Y))$ and always equals $\frac{1}{12}$, which leads to the well-known formula $\rho(X,Y) = 12 \Cov(F_X(X),F_Y(Y))$. In the discrete case, $\MCor(F_X(X),F_Y(Y))$ recovers a proposal by \cite{genest2007} to combat the non-attainability of Spearman's $\rho$ for discrete random variables, see \cite{Pohle2026} for a discussion and for classical definitions of Spearman's $\rho$ in the non-continuous case. 
\end{example}

In practice, Pearson and Spearman correlation are most often interpreted as summaries of the full dependence structure, expressed in a single number.
However, by definition covariance and Pearson correlation measure dependence around the means (which was the starting point of this paper), while Spearman's $\rho$ does the same on the rank scale, see \eqref{eq:Pearson} and \eqref{eq:Spearman}. 
Our formal approach to summary correlations, where they arise as canonical special cases, provides a powerful justification for this practical use. 
Of course, the dependence structure cannot be fully described by a single number, for example for all the bivariate copulas in Figure \ref{fig:qfcorvarious}, Spearman's $\rho$ has the same value, $\rho=0.5$, but they have very different dependence structures as the distributional correlations uncover. 
The two closely related families (threshold and quantile family) of local, distributional and canonical summary correlations provide dependence measures for different purposes and should be chosen according to the statistical problem at hand: 
measuring dependence locally, characterising the full dependence structure or condensing it in a single number.

\begin{example}[Regional measures of dependence]
	\label{example:regional_measures}
	Let us now integrate with the Lebesgue measure only over some parts of $[0,1]^2$ or $\R^2$.
	For example in the quantile function case and for $A \subset [0,1]^2$, the respective summary covariance reads:
	\begin{equation*}
	\SCov_{\QF,\lambda_A} (X,Y) :=  \int_{A}  \QCov_{\alpha,\beta} (X,Y) \d (\alpha, \beta).
	\end{equation*} 
	Considering the respective correlation $\SCor_{\QF,\lambda_A}$ and letting $A$ be located in one of the tails, e.g., $A = [0,c]^2$ for small $c$, 
	this leads to an alternative measure of tail dependence that has a similar relation to $\QCor_{c,c}$ as the expected shortfall has to the value at risk. Letting $A$ be a region in the centre of the distribution, e.g., again a rectangle $[0.5-d,0.5+d]^2$, this leads to a measure of dependence in the centre, which is in a similar relation to median covariance $\QCor_{0.5,0.5}$. The measure $\SCov_{\QF,\lambda_A} (X,Y)$ is somewhat similar to the index of interval quantile independence introduced by \citet{Zhu2018}, which amounts to integrating over the square of $\QCov_{\alpha,\beta} (X,Y)$ normalized by the Cauchy--Schwarz bound.

	When defining the corresponding quantity for the summary correlation induced by CDF correlation with a region $B \subset \mathbb{R}^2$, we get
	\begin{equation*}
	\SCov_{\CDF,\lambda_B} (X,Y) := \int_{B}  \TCov_{a,b} (X,Y) \d (a, b).
	\end{equation*}
	Here, one could focus on a whole quadrant, e.g., $B=[-\infty,0]^2$. This resembles the idea behind so-called semi-correlations \citep[Chapter 2]{Joe2014}.
\end{example}

\section{Estimation}
\label{subsec:estimation}

The empirical analogues of generalised covariances and correlations are fairly straightforward, exploiting plug-in estimators and the method of moments.
They constitute natural 
and consistent
estimators for the respective quantities on the population level.
Suppose we have a random sample $\{(X_i,Y_i),\  i=1,\ldots, n\}$ from $F_{X,Y}$ and we are interested in estimating the generalised covariance, $\Cov_{T_1,T_2}(X,Y)$, or generalised correlation, $\Cor_{T_1,T_2}(X,Y)$, at $T_1$, $T_2$.
Further, suppose that the corresponding generalised errors are 
induced by increasing and non-constant $\L_1$- and $\L_2$-identification functions $v_{1}$ and $v_{2}$, invoking Proposition \ref{prop:generalised error}.
To start with, suppose that $\L_1$ and $\L_2$ contain all random variables with any empirical distribution. 
Then, we can simply apply the definition of the generalised covariance and correlation to the empirical distribution and use this as an estimator for the population quantity.
To obtain the sample analogue $\hat t_1^n := \widehat T_1^n(X)$ 
of an identifiable functional $T_1(X)$, we use the solution in $t$ of
\be{eq:t_1^n}
\frac{1}{n}\sum_{i=1}^n v_{T_1}(t,X_i)\stackrel{!}{=}0.
\ee
Similarly, we write $\hat t_2^n := \widehat T_2^n(Y)$.
Then, the generalised covariance on the sample level is 
\be{eq:emp cov}
\widehat \Cov_{T_1,T_2}^n(X,Y) := \frac{1}{n} \sum_{i=1}^n v_{T_1}(\hat t_1^n,X_i)v_{T_2}(\hat t_2^n,Y_i).
\ee
To obtain the normalisation for the generalised correlation, we take the empirical marginal distributions from the random sample $\{(X_i,Y_i),\  i=1,\ldots, n\}$, but couple the observations with the corresponding co- and countermonotonicity copulas, utilising the increasing order statistics  $X_{(1)}\le \cdots \le X_{(n)}$ and $Y_{(1)}\le \cdots \le Y_{(n)}$.
Hence, we obtain for the comonotonic coupling $(X,Y'')$ and the countermonotonic coupling $(X,Y')$
\begin{align}
\label{eq:summation}
\widehat \Cov_{T_1,T_2}^n(X,Y'') 
&:= \frac{1}{n} \sum_{i=1}^n v_{T_1}(\hat t_1^n,X_{(i)})v_{T_2}(\hat t_2^n,Y_{(i)}), \\
\nonumber
\widehat \Cov_{T_1,T_2}^n(X,Y') 
&:= \frac{1}{n} \sum_{i=1}^n v_{T_1}(\hat t_1^n,X_{(i)})v_{T_2}(\hat t_2^n,Y_{(n-i+1)}).
\end{align}
Finally, we set 
\be{eq:emp cor}
\widehat \Cor_{\T_1, \T_2}^n(X,Y) 
:=
\begin{cases}
	\frac{\widehat \Cov_{\T_1, \T_2}^n(X,Y)}{|\widehat \Cov_{\T_1, \T_2}^n(X,Y'\,)|}, &\text{if } \widehat \Cov_{\T_1, \T_2}^n(X,Y)<0, \\[1em]
	\frac{\widehat \Cov_{\T_1, \T_2}^n(X,Y)}{|\widehat \Cov_{\T_1, \T_2}^n(X,Y''\,)|}, &\text{if } \widehat \Cov_{\T_1, \T_2}^n(X,Y)\ge0.
\end{cases}
\ee

For threshold and quantile correlation the estimation simplifies as the normalisation does not involve the co- and countermonotonic coupling. For threshold correlation, $\TCor_{a,b}(X,Y)$, we only need to estimate $F_{X,Y} (a,b)$, $F_X(a)$ and $F_Y(b)$ via   
\begin{equation}
	\label{eq:threshold correlation estimation}
\widehat F^n_{X,Y}(a,b) = \frac{1}{n} \sum_{i=1}^n \one \{ X_i \le a, Y_i \le b \},
\quad
\widehat F^n_X(a) = \frac{1}{n} \sum_{i=1}^n \one \{ X_i \le a \},
\quad
\widehat F^n_Y(b) = \frac{1}{n} \sum_{i=1}^n \one \{ Y_i \le b\}
\end{equation}
and replace the theoretical quantities by their empirical counterparts in $\TCov_{a,b}(X,Y)$ from \eqref{eq:Threshold cov} and the normalisation terms in Example \ref{example:Threshold correlation}.

To estimate quantile correlation, $\QCor_{\alpha,\beta}(X,Y)$, we first obtain the sample $\alpha$- and $\beta$-quantiles, or more formally, for the canonical identification functions $v_{q_\alpha}$ and $v_{q_\beta}$ in \eqref{eq:quantile id} we set
\[
\hat q^n_\alpha := \inf \{t \mid \frac{1}{n}\sum_{i=1}^n v_{q_\alpha}(t,X_i)\le0\},
\qquad
\hat q^n_\beta := \inf \{t \mid \frac{1}{n}\sum_{i=1}^n v_{q_\beta}(t,Y_i)\le0\}.
\]
Then, we obtain estimates for $C_{X,Y}\big(F_X(q_\alpha(X)), F_Y(q_\beta(Y))\big)$, $F_X(q_\alpha(X))$ and $F_Y(q_\beta(Y))$ by replacing the thresholds $a$, $b$ in \eqref{eq:threshold correlation estimation} with $q^n_\alpha$ and $q^n_\beta$ and replacing the theoretical quantities in \eqref{eq:Quantile correction copula} and the normalisation terms in Example \ref{example:Quantile correlation} with those estimators.

To construct estimators for distributional correlations, we use the respective estimators for the local correlations just described on a grid (of size 10000 in the following applications).  We close this section by establishing the consistency of the described estimators, subject to typical regularity conditions.

\begin{proposition}
	\label{thm:consistency}
	Suppose that for $i=1,2$ the identification functions $v_i$ are increasing, \emph{strict} $\L_i$-identification functions for $T_i$, such that 
	the families $v_1(\cdot,x)_{x\in\R}$ and $v_2(\cdot, y)_{y\in \R}$ are pointwise equicontinuous and let $X\in\L_1, Y \in \L_2$.
	Further, suppose that $\{F_X: X\in\L_1\}$ and $\{F_Y:Y\in\L_2\}$ contain all empirical distribution functions, and suppose that there exist $p,q\in[1,\infty]$ with $p^{-1}+q^{-1}<1$ such that $e_{T_1}(X)\in L^p(\R)$ and $e_{T_2}(Y)\in L^q(\R)$.\\
	Then, the estimators $\widehat \Cov_{T_1,T_2}^n(X,Y)$ \eqref{eq:emp cov} and $\widehat \Cor_{\T_1, \T_2}^n(X,Y)$ \eqref{eq:emp cor} based on a random sample are strongly consistent. That is, they converge almost surely to $\Cov_{T_1,T_2}(X,Y)$ and $\Cor_{T_1,T_2}(X,Y)$, respectively.
\end{proposition}

Threshold covariance and threshold correlation do not satisfy the conditions of Proposition \ref{thm:consistency}, but the strong consistency follows directly from the strong law of large numbers and the continuous mapping theorem.
For quantile covariance and quantile correlation, we establish strong consistency if the marginal distributions are continuous at the respective quantiles, see Proposition \ref{prop:consistency quantiles}.

Beyond the consistency results discussed above, uncertainty quantification of estimators constitutes an important statistical task, enabling the construction of confidence regions and tests.
This requires knowledge about the (asymptotic) distribution of our estimators. 
The case distinction in the normalisation of \eqref{eq:emp cor}, however, implies non-standard asymptotic distributions, in particular asymmetric Gaussian distributions. 
Hence, a detailed discussion of uncertainty quantification is beyond the scope of this paper and is relegated to future research.
We point at 
\citet{Pohle2024} who derive the asymptotic distribution of estimators for mean correlation for binary variables, which can serve as the basis for the treatment of quantile, threshold and other generalised correlations. 
Further, while pointwise inference for our function-valued distributional correlations would be enabled by the limiting distributions for quantile and threshold correlation, joint inference could be tackled via the close relationship to the copula and the joint CDF, exploiting existing work on empirical (copula) processes (see \citet{buecher2013} and references therein).

\section{Data examples}
\label{subsec:data examples}

To illustrate the use of local, distributional and summary correlations in practice, we extract data on mixed-sex couples from the 2019 wave of the 
\cite{PSID2023}.
After data cleaning we have a sample of 4417 couples living in the same households, of whom 85\% are married. 
As all of the variables we analyse are discrete (either due to being inherently discrete or, e.g., heights being recorded in full inches), we use bubble plots for our scatter plots to avoid overplotting. 
Further, note that CDF and quantile function correlation only change their value at jumps of the CDF (see Remark \ref{rem:last jump 2}). 
In all graphical representations we focus on the ranges between the 2.5\%- and the 97.5\%-quantiles.

\begin{figure}
	\centering
	\includegraphics[width=1\linewidth]{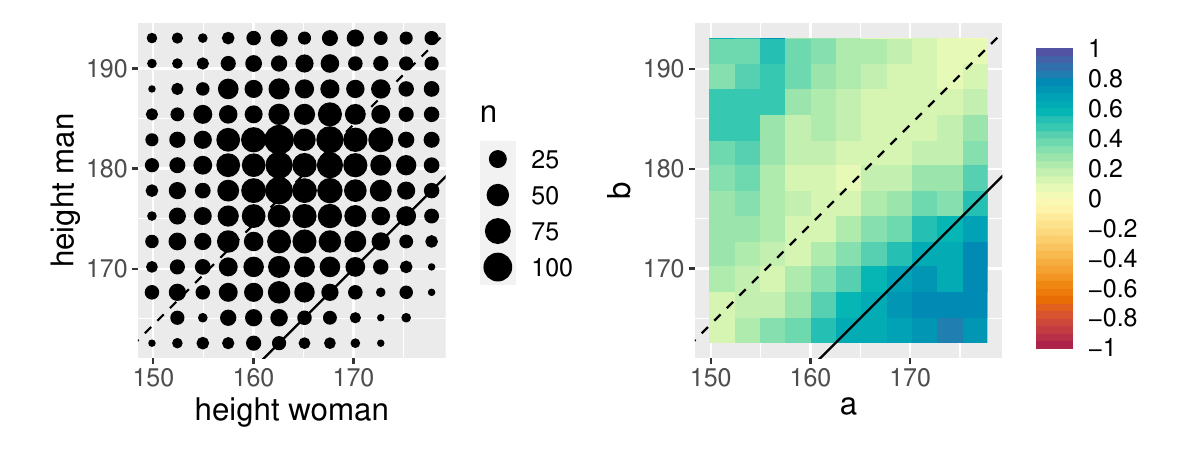}
	\caption{Bubble plot and empirical $\CDFCor$ of heights of mixed-sex couples, solid line: diagonal; dotted line: diagonal plus average height difference}
	\label{fig:cdfcorheightsplusscatterpluslines}
\end{figure}

We first consider the heights of the couples in cm. Figure \ref{fig:cdfcorheightsplusscatterpluslines} depicts a scatter plot and the corresponding CDF correlation. 
Heights of men and women are globally positively dependent as CDF correlation is positive everywhere, indicating a preference for assortative mating \citep{stulp2013}. 
Mean correlation equals $\MCor=0.216$ (Spearman's $\rho=0.189$), indicating a weak positive relation on average. 
The local dependence structure, however, varies strongly: 
In the lower right corner, the dependence is quite strong, while elsewhere (with the exception of the upper left corner)
it is weak. 
This reflects the male-taller norm in Western societies \citep{stulp2013}: 
For women and men fairly close to the average height difference (indicated by the dotted line), CDF correlation is close to 0, indicating that mating behaviour in this region is hardly influenced by the partners' height (in relation to their own). 
As the height of men approaches the height of women (the diagonal is represented by the solid line), CDF correlation rises abruptly, suggesting that the mating behaviour is strongly influenced by height in this region: 
heights of partners are strongly positively associated there. (Actually, the norm is rather that the female partner should be at least a few centimetres smaller than the male partner.) 
When computing the regional summary correlation $\SCor_{\CDF,\lambda_B}$ from Example \ref{example:regional_measures} below and above the diagonal we consequently get 0.601 and 0.201, respectively. 
In the far upper left corner the dependence gets quite strong as well, reflecting the male-not-too-tall norm \citep{stulp2013}. 
Quantile function correlation yields qualitatively the same picture, see Figure \ref{fig:qfcorheights} in the Appendix and the discussion below.

\begin{figure}[ht]
	\centering
	\includegraphics[width=1\linewidth]{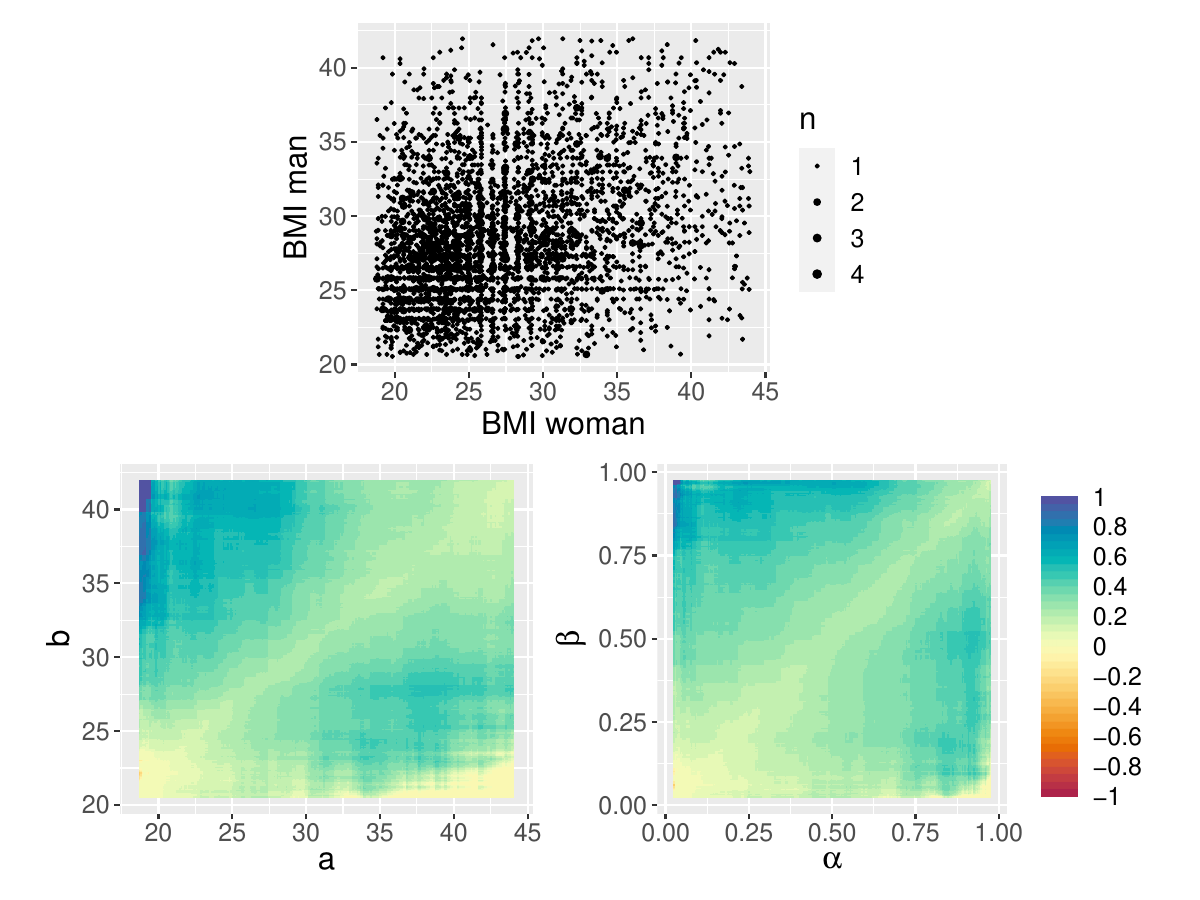}
	\caption{Bubble plot and empirical $\CDFCor$ and $\QFCor$ of BMIs of mixed-sex couples.}
	\label{fig:cdfcorqfcorbmisplusscatter}
\end{figure}

Figure \ref{fig:cdfcorqfcorbmisplusscatter} depicts the bubble plot, the CDF and quantile function correlation for the body mass indices (BMI, mass in kg divided by the square of height in cm) of the couples. 
Again the two variables are (almost) globally positively dependent and the canonical summary correlations from examples \ref{example:Summary_mean_correlation} and \ref{example:Spearman_correlation} take the values $\MCor=0.304$ and $\rho=0.287$. 
In the lower left corner below a BMI of 25, which corresponds to the threshold between having a normal weight and being overweight according to the 
\cite{WHO2010}, $\CDFCor$ and $\QCor$ are close to 0, whereas elsewhere they are larger, in particular if one of the partners exceeds the threshold of 30, from which on a person is classified as obese according to the \cite{WHO2010}. 
Thus, if both partners are in a normal weight range, there seems to be no association between BMIs, but once one of the partners is obese, the association gets strongly positive. Here, the dependence is of course not only generated by mating behaviour, but also by mutual influence in lifestyle.
This example nicely illustrates how CDF and quantile correlation may differ and how they complement each other: 
While $\CDFCor$ is often nicely interpretable as it operates on the observation scale itself, sometimes a lot of space in the plot is occupied by regions with only few observations (which is why we chose to present only the central regions of the axes with 95\% of the probability mass in the first place). 
For example here, due to the marginal distributions of the BMIs being right-skewed, regions with higher values occupy comparatively more space in the plot. 
$\QFCor$, on the other hand, always has a solid interpretation in terms of quantile levels and naturally assigns space in the plot according to probability mass.

\begin{figure}[ht]
	
	\includegraphics[width=1\linewidth]{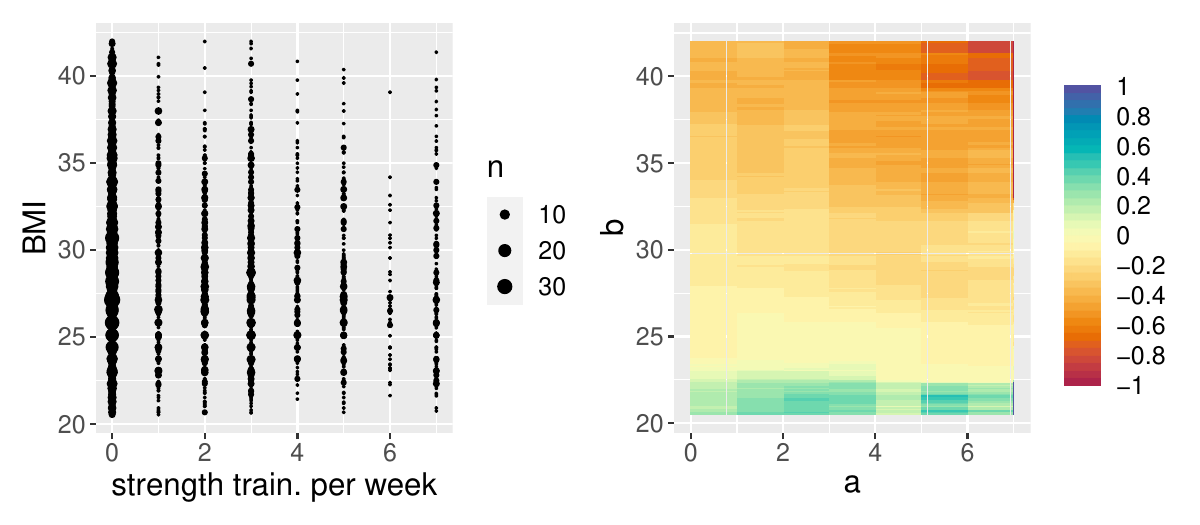}
	\caption{Bubble plot and empirical $\CDFCor$ of the number of strength trainings per week and BMI of the men.}
	\label{fig:cdfcorstrengthtrainingbmiplusscatter}
	
\end{figure}

Next, we consider an example exhibiting regions of positive and of negative dependence. 
Figure \ref{fig:cdfcorstrengthtrainingbmiplusscatter} contains the scatter plot as well as the CDF correlation of the number of strength trainings per week and the BMI of the men in our sample. 
On average, frequency of strength training is negatively associated with BMI, $\MCor=-0.112$ and $\rho=-0.088$. 
Locally, the negative dependence starts to arise with about the overweight threshold of 25 and gets stronger the higher BMI gets, being particularly strong for people above the obesity threshold of 30. 
For men in the upper normal range of BMI, $\CDFCor$ is close to 0, whereas for men in the lower normal range, the association is positive. 
Thus, for men with a low BMI, strength training is associated with a gain in body weight -- probably due to increased muscle mass, while for overweight males it is associated with a loss in body weight -- probably due to fat loss outweighing increased muscle mass.  

In Part \ref{app:data examples} of the Appendix we provide the above-mentioned additional two figures and compare mean correlation from Example \ref{ex:Mean correlation} and Pearson correlation from \eqref{eq:Pearson} for the three examples above and a fourth one.

\section{Conclusion}
\label{sec:conclusion}

We present new concepts and measures of dependence.
On the one hand, we consider dependence from the perspective of statistical functionals and put forward generalised covariances and correlations as corresponding dependence measures. 
On the other hand, with our local and distributional correlations we introduce local dependence (including tail dependence) measures as well as function-valued measures uncovering the full dependence structure. 
Summary correlations average over distributional correlations and close the loop to classical measures of dependence like covariance, Pearson correlation and Spearman's $\rho$. 
We analyse the properties of the new measures and present first applications.

These measures open many opportunities for future research. 
First of all, they will be useful in a wide variety of applications, possibly providing deeper insights about dependence structures than classical measures. 
Our measures are concerned with dependence between two random variables. 
Thus, they naturally extend to settings where pairwise dependence plays a role, for example correlation matrices, temporal or spatial dependence. 
Extending our measures to examine dependence for a vector of variables jointly is naturally more difficult, as negative dependence turns into a subtle concept for more than two variables \citep[Chapter 3.3]{Mari2001}. 
As discussed at the end of Section \ref{subsec:estimation}, uncertainty quantification of our estimators requires non-standard asymptotic distributions, and is therefore relegated to future research. Our measures are fundamentally different from recent popular measures of functional dependence \citep{szekely2007, szekely2009, chatterjee2021, reshef2011}, which are not measures of directional dependence, but only of strength of dependence, thus mapping to $[0,1]$. 
Analysing possible connections of those measures to distributional or summary correlations is certainly interesting. 
Finally, exploring the links between generalised correlations and generalised regression approaches might be fruitful.

\section*{Acknowledgements}

We are grateful to Patrick Cheridito, Holger Dette, Cees Diks, Timo Dimitriadis, Jan Gertheiss, Tilmann Gneiting, Bettina Gr\"un, Marius Hofert, Alexander Jordan, Alexander McNeil, Johanna Ne\v{s}lehov\'a, Christoph Rothe, Melanie Schienle, Christian H.\ Wei{\ss} and Jan-Lukas Wermuth for valuable discussions about the topic. We further thank seminar participants at Helmut Schmidt University Hamburg, Heidelberg University, Heidelberg Institute for Theoretical Studies, Vienna University of Economics and Business, University of Sussex, and ETH Zurich, and conference participants at COMPSTAT 2024, the 11th Bernoulli-IMS World Congress in Probability and Statistics, the 10th HKMetrics-Workshop, DAGStat 2022 and the Bernoulli Young Researcher Event 2022 for helpful comments. 
Marc-Oliver Pohle is grateful for support by the Klaus Tschira Foundation, Germany. 

The collection of data used in this study was partly supported by the National Institutes of Health under grant number R01 HD069609 and R01 AG040213, and the National Science Foundation under award numbers SES 1157698 and 1623684.

\bibliographystyle{apalike}
\bibliography{library_correlations}

\newpage	
	
\appendix
\appendixpage

\section{Proofs and technical results}
\label{sec:proofs}

	\begin{proof}[Proposition \ref{prop:generalised error}]
		We show that the three properties of a generalised error from Definition \ref{def:generalised error} are satisfied.
		Property \eqref{prop:centred} follows from the fact that $T$ is an $\L$-identification function. 
		Property \eqref{prop:increasing} is an immediate consequence from the assumption that $v_T$ is increasing.
		To verify \eqref{prop:sign change}, it suffices to show that $v_{\T}(t,x)(x-t)\ge0$ for all $(t,x)\in\A\times \R$.
		Suppose this is violated by some $x_0\neq t_0$. 
		If $t_0 < x_0$, this means that 
		$v_{\T}(t_0,x_0)<0$. 
		Since $v_{\T}$ is increasing it holds that $v_{\T}(t_0,x)<0$ for all $x\le x_0$. 
		Due to Assumption \ref{ass:T} there is some $X_0\in\L$ such that $T(X_0)=t_0$ and $\esssup(X_0)=x_0$. 
		Then clearly $\E[v(t_0,X_0)]<0$, which violates the fact that $v_{\T}$ is an $\L$-identification function for $\T$.
		The case $t_0 > x_0$ works analogously. 
	\end{proof}

\begin{proof}[Lemma \ref{lem:symmetric}]
	The ``if'' direction is obvious. 
	For the ``only if'' direction, we assume attainability. Using \citet[Theorem 7.28]{mcneil2015}, we know that $Y$ is of the same type as $X$ and as $-X$. So we just need to show that $X$ and $Y$ have symmetric distributions. 
We know that there are $a,a'>0$ and $b,b'\in\R$ such that $Y\stackrel{\mathrm{d}}{=} aX + b$ and $Y\stackrel{\mathrm{d}}{=} -a'X + b'$.
	This implies $aX + b \stackrel{\mathrm{d}}{=} -a'X + b'$, which implies $X - c \stackrel{\mathrm{d}}{=} -\lambda (X - c)$ for $c= \frac{b'-b}{a'+a}\in\R$ and $\lambda=\frac{a'}{a}>0$.	
	Hence, using Lemma \ref{lem:auxiliary} and the fact that $X$ is non-constant, this implies the symmetry of the distribution of $X$. Since $Y$ is of the same type, it must also be symmetric.
\end{proof}

\begin{lemma}\label{lem:auxiliary}
	For any $X\in L^0(\R)$ we have the following implication.
	If $X\stackrel{\mathrm{d}}{=}-\lambda X$ for some $\lambda>0$, then $\lambda=1$ or $X=0$ almost surely.
\end{lemma}

\begin{proof}[Lemma \ref{lem:auxiliary}]
Suppose that $X \stackrel{(d)}{=}-\lambda X$ for some $\lambda>0$.
W.l.o.g., we can assume that $\lambda\ge 1$. Then, for any $a>0$ we have
\(
\mathbb P\big[X\in [-a,a]\big] = \mathbb P\big[X\in [-\lambda^{-1}a,\lambda^{-1}a]\big].
\)
Consequently, for any $a>0$ it holds that
\(
\mathbb P\big[X\in [-a,-\lambda^{-1}a)\big] = \mathbb P\big[X\in (\lambda^{-1}a,a]\big]=0.
\)
Therefore, if $\lambda>1$ then $\mathbb P\big[X\in \R\setminus\{0\}\big]=0$, which implies that $X=0$ almost surely. Otherwise, $\lambda =1$.
\end{proof}

\begin{proof}[Proposition \ref{prop:Frechet-Hoeffding}]
	The proof largely follows the arguments of Theorem 4 in \cite{Embrechts2002}.
	Since $e_{\T_1}(X)$ and $e_{\T_2}(Y)$ are increasing functions of $X$ and $Y$, respectively, the pair $\big(e_{\T_1}(X), e_{\T_2}(Y')\big)$ is countermonotonic and $\big(e_{\T_1}(X), e_{\T_2}(Y'')\big)$ is comonotonic.
	\begin{enumerate}[(i)]
		\item
		Let $F$ be the distribution function of $\big(e_{\T_1}(X), e_{\T_2}(Y)\big)$ with marginals $F_1, F_2$.
		Due to the formula of Hoeffding \cite[Lemma 7.27]{mcneil2015}, it holds that 
		\begin{align*}
			\Cov_{\T_1,\T_2}(X,Y')&= \iint \max \big(F_1(z_1) + F_2(z_2) - 1,0\big) - F_1(z_1)F_2(z_2)\d z_1 \mathrm{d} z_2, \\
			\Cov_{\T_1,\T_2}(X,Y'')&= \iint \min \big(F_1(z_1), F_2(z_2)\big) - F_1(z_1)F_2(z_2)\d z_1 \mathrm{d} z_2,
		\end{align*}
		using the identity \eqref{eq:cov identity}.
		The Fr\'echet bounds imply that for any $z_1,z_2\in\R$
		\[
		\max \big(F_1(z_1) + F_2(z_2) - 1,0\big)\le F(z_1,z_2) \le \min \big(F_1(z_1), F_2(z_2)\big).
		\]
		Hence, the integral 
		\[
		\iint F(z_1,z_2) - F_1(z_1)F_2(z_2)\d z_1 \mathrm{d} z_2
		\]
		exists, is finite, and coincides with $\Cov_{\T_1,\T_2}(X,Y)$.
		\item
		This follows along the lines of part (2) in \citet[Theorem 4]{Embrechts2002}.
		\item
		Let $g_1,g_2$ be strictly increasing functions such that $\P$-almost surely $g_1(X) = e_{\T_1}(X)$ and $g_2(Y) = e_{\T_2}(Y)$.
		Part (2) of \cite[Theorem 4]{Embrechts2002} asserts that an equality in \eqref{eq:Frechet-Hoeffding} is attained only if $\big(e_{\T_1}(X), e_{\T_2}(Y)\big)$ is co- or countermonotonic.
		Suppose the second equality in \eqref{eq:Frechet-Hoeffding} is attained such that $\big(e_{\T_1}(X), e_{\T_2}(Y)\big)$ is comonotonic. That means there are increasing functions $\nu_1,\nu_2$ and a random variable $Z$ such that 
		$\big(g_1(X),g_2(Y)\big) \dist \big(\nu_1(Z), \nu_2(Z)\big)$. Since $g_1$ and $g_2$ are strictly increasing, they can be inverted on their respective images.
		This yields that $(X,Y) \dist \big(g_1^{-1}\circ \nu_1(Z), g_2^{-1}\circ\nu_2(Z)\big)$. Since $g_1^{-1}$ and $g_2^{-1}$ are increasing, $(X,Y)$ is comonotonic.
	\end{enumerate}
\end{proof}

\begin{proof}[Proposition \ref{prop:nullity}]
	For $\mathrm{D}$ being $\TCov$ (and likewise $\TCor$), this is immediate due to \eqref{eq:Threshold cov}.
	For $\QCov$ (and likewise $\QCov$), the ``if'' direction is obvious. For the ``only if'' direction, it suffices to consider $\QCov$ and to invoke \eqref{eq:Quantile correction copula}. It suffices to show that $C_{X,Y}$ corresponds to the independence copula on $\text{range}(F_X)\setminus\{0,1\} \times \text{range}(F_Y)\setminus\{0,1\}$. Since the image of $(0,1)\ni\alpha\mapsto F_X(q_\alpha(X))$ is a superset of $\text{range}(F_X)\setminus\{0,1\}$, the claim follows.
\end{proof}

\begin{proof}[Corollary \ref{cor:distributional properties}]
	(i) follows directly from Theorem \ref{theorem:properties} (i) and the ``if'' direction of (ii) from \ref{theorem:properties} (iii). The ``only if'' direction of (ii) follows by the representations \eqref{eq:Threshold cov} and \eqref{eq:Quantile correction copula}, the normalisations from examples \ref{example:Threshold correlation} and \ref{example:Quantile correlation} and the fact that if $X$ and $Y$ have the upper (lower) Fr\'echet--Hoeffding bound as a copula, they are comonotonic (countermonotonic), see \citet[propositions 7.18 and 7.22]{mcneil2015}.
\end{proof}

\begin{proof}[Proposition \ref{prop:global dependence}]
	Proposition \ref{prop:dependence} implies the ``only if'' direction. The ``if'' direction follows from the representation \eqref{eq:Quantile correction copula} and Sklar's theorem.
\end{proof}

\begin{proof}[Proposition \ref{prop:dependence}]
	Assume that $X$ and $Y$ are globally positively dependent. Increasing functions of globally positively dependent variables are positively dependent \cite[Lemma 1 (iii)]{lehmann1966}. This implies that the generalised errors $e_{T_1}(X)$ and $e_{T_2}(Y)$ are globally positively dependent. 
	As global positive dependence implies non-negative Pearson correlation \cite[Lemma 3]{lehmann1966}, 
	\eqref{eq:cov identity} implies $\Cov_{T_1,T_2}(X,Y) \geq 0$. 
	For global negative dependence, the same reasoning applies.
\end{proof}

\begin{proof}[Corollary \ref{cor:summary correlations}]
	For (i) note that $\QCov_{\alpha,\beta} (X,Y)$ and $\TCov_{a,b} (X,Y)$ are bounded by the Fr\'echet--Hoeffding bounds from Examples \ref{example:Quantile correlation} and \ref{example:Threshold correlation}. 
	(ii) follows directly from Proposition \ref{prop:nullity}. (iii) follows from Corollary \ref{cor:distributional properties} (ii) and (iv) from the respective part of Theorem \ref{theorem:properties}.
\end{proof}

\section{Details on the uniqueness of generalised correlations } \label{app:Generalised errors and identification functions}

	Here, we show that if one uses a generalised error induced by any identification function for a functional, this leads to the same generalised correlation, regardless of the choice of the identification function. 
	In this sense, generalised correlation for a functional is unique when using this construction principle.

\begin{remark}
\label{rem:dependence of id}	
	Obviously, a generalised $T_1-T_2$-covariance is not unique as it depends on the choices of the generalised errors for $T_1$ and $T_2$. However, when assuming that the errors are induced by the increasing identification functions $v_{T_1}$ and $v_{T_2}$, a recent characterisation result in \citet[Theorem 4]{Osband_Id} helps to get an understanding of how the choice of these identification functions influences the generalised covariance.
	It implies that -- under richness conditions on $\L$ and further regularity conditions (Richness essentially means that $T\colon\L\to\A$ is surjective (readily implied by Assumption \ref{ass:T} (b)) and $\{F_X\,|\,X\in\L\}$ is convex.) -- $v$ and $v'$ are two increasing and non-constant $\L$-identification functions for $T$ if and only if there is a positive function $h\colon\A\to\R$ such that 
	\begin{equation}
		\label{eq:Osband identification}
		v'(t,x) = h(t)v(t,x) \qquad \text{for all }x\in\R, \ t\in\A.
	\end{equation}
	This implies that two generalised covariances constructed from identification functions only differ in a factor depending on the two functionals $T_1(X)$ and $T_2(Y)$. Consequently, the	dependence on the choice of the identification function vanishes when considering generalised correlations.	
\end{remark}

	\begin{proposition}
		\label{prop:ind from h}
		Under the conditions of \citet[Theorem 4]{Osband_Id}, the generalised correlation at $\T_1$ and $\T_2$ constructed on generalised errors induced by  $v_1$ and $v_2$ via \eqref{eq:error} does not depend on the choice of the increasing and non-constant identification functions.
	\end{proposition}
	
	\begin{proof}
		Let $v_1', v_2'$ be two other increasing and non-constant identification functions for $T_1$ and $T_2$, respectively. According to \citet[Theorem 4]{Osband_Id}, see also Equation \eqref{eq:Osband identification}, there are two positive functions $h_1, h_2$ such that $v'_1 (t,x) = h_1(t)v_1(t,x)$ and $v'_2(t,x) = h_2(t)v_2(t,x)$.
		Then for $Z\in\{Y,Y',Y''\}$, we obtain for the covariance induced by $v'_1, v'_2$
		\[
		\Cov'_{T_1,T_2}(X,Z) = h_1\big(T_1(X)\big)h_2\big(T_2(Y)\big) \Cov_{T_1,T_2}(X,Z),
		\]
		where $\Cov_{T_1,T_2}(X,Z)$ is the generalised covariance induced by $v_1,v_2$. Since $h_1$ and $h_2$ are positive and since $T_2(Y) = T_2(Y') = T_2(Y'')$, \eqref{eq:generalised cor} directly yields the claim.
	\end{proof}

\section{Consistency of the estimators}

For the proof of Proposition \ref{thm:consistency}, we first need to establish the following proposition.

\begin{proposition}\label{prop:comonotonic LLN}
Let $(X,Y)$ be a random vector with distribution $F_{X,Y}$ and let $\{(X_i,Y_i), \ i=1,\ldots, n\}$ be a random sample from this distribution. Let $(X,Y'')$ be the comonotonic and $(X,Y')$ the countermonotonic coupling of $(X,Y)$.
\begin{enumerate}[(i)]
\item
The empirical distribution induced by $\{(X_{(i)},Y_{(i)}), \ i=1,\ldots, n\}$ converges almost surely uniformly to $F_{X,Y''}$. Similarly, the empirical distribution induced by $\{(X_{(i)},Y_{(n-i+1)}), \ i=1,\ldots, n\}$ converges almost surely uniformly to $F_{X,Y'}$.
\item
If $X\in L^p(\R)$ and $Y\in L^q(\R)$ with $p,q\in[1,\infty]$ such that $p^{-1}+q^{-1}<1$. 
Then the expectations $\E[XY]$, $\E[XY']$ and $\E[XY'']$ are finite (due to H\"older's inequality) and we have the following almost sure convergence:
\[
\frac{1}{n}\sum_{i=1}^n X_{(i)}Y_{(i)} \to \E[XY''] \qquad \text{and}
\qquad 
\frac{1}{n}\sum_{i=1}^n X_{(i)}Y_{(n-i+1)} \to \E[XY']\,.
\]
\end{enumerate}

\end{proposition}

\begin{proof}
Without loss of generality, we only prove the comonotonic case. For part (i), consider the empirical distribution induced by $\{(X_{(i)}, Y_{(i)}), \ i=1, \ldots, n\}$:
\[
	\widehat G_n(x,y) := \frac{1}{n}\sum_{i=1}^n \one\{X_{(i)}\le x\} \one\{Y_{(i)}\le y\}, \qquad x,y\in\R\,.
	\]
	The marginal distributions of $\widehat G_n$ are 
	\begin{align*}
	\widehat G_{1,n}(x) &= \frac{1}{n}\sum_{i=1}^n \one\{X_{(i)}\le x\} = \frac{1}{n}\sum_{i=1}^n \one\{X_{i}\le x\}=\widehat F_{1,n}(x), \\
	\widehat G_{2,n}(y) &= \frac{1}{n}\sum_{i=1}^n \one\{Y_{(i)}\le y\} = \frac{1}{n}\sum_{i=1}^n \one\{Y_{i}\le y\}= \widehat F_{2,n}(y)\,.
	\end{align*}
	for $x,y\in\R$.
	
	Thanks to reordering the sums, we obtain that the marginal distributions of  $\widehat G_n$ correspond to the empirical distributions $\widehat F_{1,n}$ and $\widehat F_{2,n}$, induced by $\{X_i, \ i=1, \ldots, n\}$ and $\{Y_i, \ i=1, \ldots, n\}$, respectively.
	The copula of $\widehat G_n$ is not unique, since its marginal distributions are not continuous. But it is easy to see that the comonotonicity copula $M(u_1,u_2) = \min(u_1,u_2)$, $u_1,u_2\in[0,1]$, is a possible copula of $\widehat G_n$. 
	Therefore, using Sklar's Theorem, we can write 
	\[
	\widehat G_n(x,y) = M\big(\widehat F_{1,n}(x),\widehat F_{2,n}(y)\big), \qquad x,y\in\R\,.
	\]
	Thanks to the Glivenko--Cantelli Theorem, $\widehat F_{1,n}$ converges uniformly almost surely to $F_X$. And so converges $\widehat F_{2,n}$ to $F_Y$. 
	Since $M$ is uniformly continuous, $\widehat G_n$ converges uniformly almost surely to $M\big(F_X(\cdot),F_Y(\cdot)\big) = F_{X,Y''}$.
	
	For part (ii), we can first use part (i), which of course implies that the empirical distribution $\widehat G_n$ converges weakly almost surely to $ F_{X,Y''}$. \\
	Due to the continuous mapping theorem, also the empirical distribution induced by the product $\{X_{(i)} Y_{(i)}, \ i=1, \ldots, n\}$ converges weakly almost surely to the distribution of the product $XY''$. Under weak convergence, the means $n^{-1} \sum X_{(i)} Y_{(i)}$ converge almost surely to $\E[XY'']$ if and only if the uniform integrability condition
	\begin{equation}\label{eq:uniform int}
	\lim_{K\to\infty} \limsup_{n\to\infty} \frac{1}{n} \sum_{i=1}^n Z_i \one\{Z_i > K\} =0, \qquad \text{where}\quad Z_i=|X_{(i)} Y_{(i)}|
	\end{equation}
	holds almost surely; see e.g.\ Lemma 5.11 in \cite{Kallenberg2021}. 
	Since $p^{-1}+q^{-1}<1$, there is an $r>1$ such that 
	\begin{align*}
	 \limsup_{n\to\infty} \frac{1}{n} \sum_{i=1}^n |Z_i|^r
	 \le  \limsup_{n\to\infty}\left( \frac{1}{n} \sum_{i=1}^n |X_i|^p\right)^{\frac{1}{p}} \left( \frac{1}{n} \sum_{i=1}^n |Y_i|^q\right)^{\frac{1}{q}} = \big(\E|X|^p\big)^{\frac{1}{p}} \big(\E|Y|^q\big)^{\frac{1}{q}}\,.
	\end{align*}
	Here, we exploit H\"older's inequality, reordering of finite sums and the strong law of large numbers applied to the sequences $|X_i|^p$ and $|Y_i|^q$. 
	Now, the uniform integrability \eqref{eq:uniform int} follows as described in \cite{Kallenberg2021}, directly before Lemma 5.10.
\end{proof}

\begin{proof}[Proposition \ref{thm:consistency}]
	Let $\{(X_i,Y_i),\  i=1,\ldots, n\}$ be a random sample from $F_{X,Y}$. 
	Denote the sample estimators for $T_1(X)$ and $T_2(Y)$ by $\hat t_1^n$ and $\hat t_2^n$, defined around \eqref{eq:t_1^n}.
	Due to \citet[Corollary 3.2]{HuberRonchetti2009}, the estimators $\hat t_1^n$ and $\hat t_2^n$ converge to $T_1(X)$ and $T_2(Y)$ almost surely.
	Due to the equicontinuity, it holds that for all $\eps>0$ and for $\mathbb P$-almost all $\omega\in\Omega$ there exists an $N\in\mathbb N$ such that for all $n\ge  N$ and for all $i\in\mathbb N$
	\be{eq:first convergence}
	\Big| v_1\big(\hat t_1^n(\omega),X_i(\omega)\big)v_2\big(\hat t_2^n(\omega),Y_i(\omega)\big) - v_1\big(T_1(X),X_i(\omega)\big)v_2\big(T_2(Y),Y_i(\omega)\big)\Big|<\eps \,.
	\ee
	Hence, for all $\eps>0$ and for $\mathbb P$-almost all $\omega\in\Omega$ there exists an $N\in\mathbb N$ such that for all $n\ge  N$
	\[
	\Big| \frac{1}{n}\sum_{i=1}^n  v_1\big(\hat t_1^n(\omega),X_i(\omega)\big)v_2\big(\hat t_2^n(\omega),Y_i(\omega)\big) 
	- \frac{1}{n}\sum_{i=1}^n  v_1\big(T_1(X),X_i(\omega)\big)v_2\big(T_2(Y),Y_i(\omega)\big)\Big|<\eps \,.
	\]
	Moreover, by the strong law of large numbers, for all $\eps>0$ and for $\mathbb P$-almost all $\omega\in\Omega$ there exists an $N\in\mathbb N$ such that for all $n\ge  N$ 
	\[
	\Big| \frac{1}{n}\sum_{i=1}^n  v_1\big(T_1(X),X_i(\omega)\big)v_2\big(T_2(Y),Y_i(\omega)\big)
	-
	\E\big[ v_1\big(T_1(X),X\big)v_2\big(T_2(Y),Y\big) \big]\Big|<\eps\,.
	\]
	This establishes the strong consistency of $\widehat \Cov_{T_1,T_2}^n(X,Y)$.
	
	For the strong consistency of $\widehat \Cov_{T_1,T_2}^n(X,Y'')$, we apply the steps as above and Proposition \ref{prop:comonotonic LLN} for the almost sure convergence
	\[
	\frac{1}{n}\sum_{i=1}^n v_1(T_1(X),X_{(i)}) v_2(T_2(Y),Y_{(i)}) \to \E\big[ v_1\big(T_1(X),X\big)v_2\big(T_2(Y),Y''\big)\,.
	\]
	Note that the generalised errors $e_{T_1}(X) = v_1(T_1(X),X)$ and $e_{T_2}(Y) = v_2(T_2(Y),Y)$ are increasing functions of $X$ and $Y$, respectively. The strong consistency of $\widehat \Cov_{T_1,T_2}^n(X,Y')$ follows similarly such that the strong consistency of $\widehat \Cor_{T_1,T_2}^n(X,Y)$ can be established with the continuous mapping theorem.
\end{proof}	

Clearly, when dealing with quantiles, we are not in the situation of Theorem \ref{thm:consistency} since the identification functions fail to be continuous in their first arguments. 
The following proposition is an alternative result, which can be straightforwardly applied to quantile covariance in the case when the marginals $F_{X}$ and $F_{Y}$ are continuous at their $\alpha$- and $\beta$-quantile, respectively.

\begin{proposition}
	\label{prop:consistency quantiles}
	Let $X,Y\in L_0(\R)$, such that $F_X$ and $F_Y$ are continuous at their $\alpha$- and $\beta$-quantile, respectively. Then the estimators for the quantile covariance $\QCov_{\alpha, \beta}(X,Y)$ \eqref{eq:quantile copula} and quantile correlation (Example \ref{example:Quantile correlation}) are strongly consistent.
\end{proposition}

\begin{proof}
	The proof works similar to the one of Theorem \ref{thm:consistency}. Note that the convergence result at \eqref{eq:first convergence} still holds since for all $i\in\mathbb N$ it holds that $X_i$ and $Y_i$ are almost surely different from $q_\alpha(X)$ and $q_\beta(Y)$ by assumption. 
\end{proof}

\section{Additional figures for the theoretical examples}	
\label{app:theoretical examples}

Figures \ref{fig:qfcorcauchyCS} to \ref{fig:qfcorvariousCS} complement the figures for the theoretical examples in the main text and are discussed there. 

\begin{figure}[ht]
	\centering
	\includegraphics[width=1\linewidth]{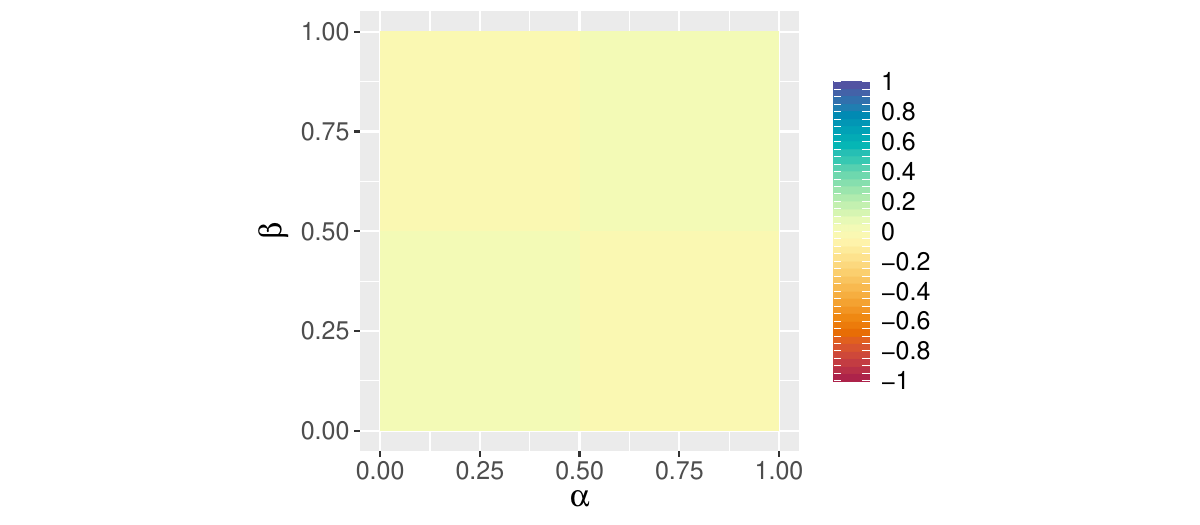}
	\caption{Plot of $\QFCov$ normalised via the Cauchy--Schwarz inequality, see \eqref{eq:CSI 3}, for a bivariate Cauchy copula with Spearman's $\rho$ equal to 0. Remark: The function is bounded in absolute value by 0.05, which leads to its graphical representation being constant within the quadrants since we use 40 bins for the colour scale.}
	\label{fig:qfcorcauchyCS}
\end{figure}

\begin{figure}[ht]
	\centering
	\includegraphics[width=1\linewidth]{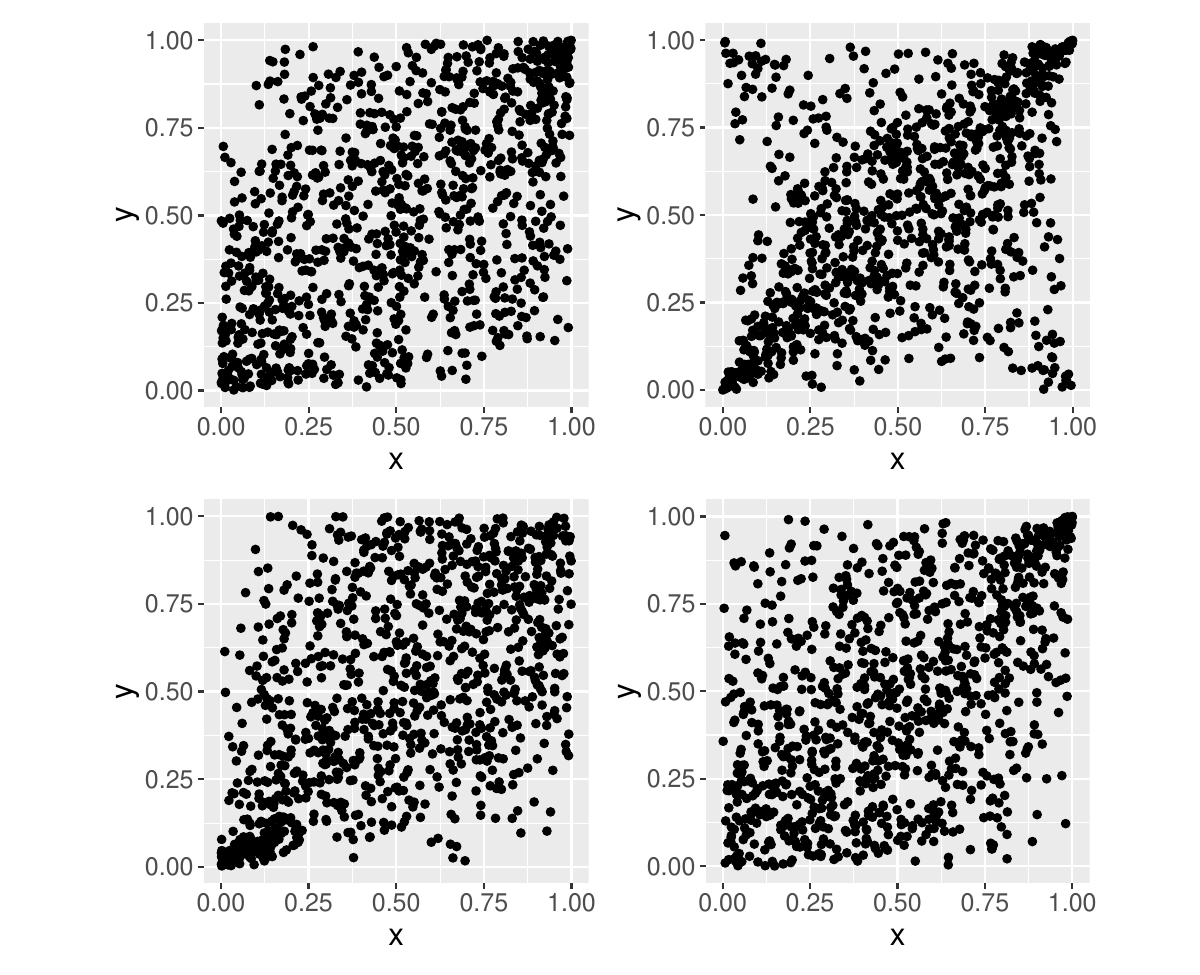}
	\caption{Scatter plots of 1000 random draws from four different copulas, all with Spearman's $\rho$ equal to 0.5: Gaussian (upper left), Cauchy (upper right), Clayton (lower left) and Gumbel (lower right)}
	\label{fig:qfcorvariousscatter}
\end{figure}

\begin{figure}[ht]
	\centering
	\includegraphics[width=1\linewidth]{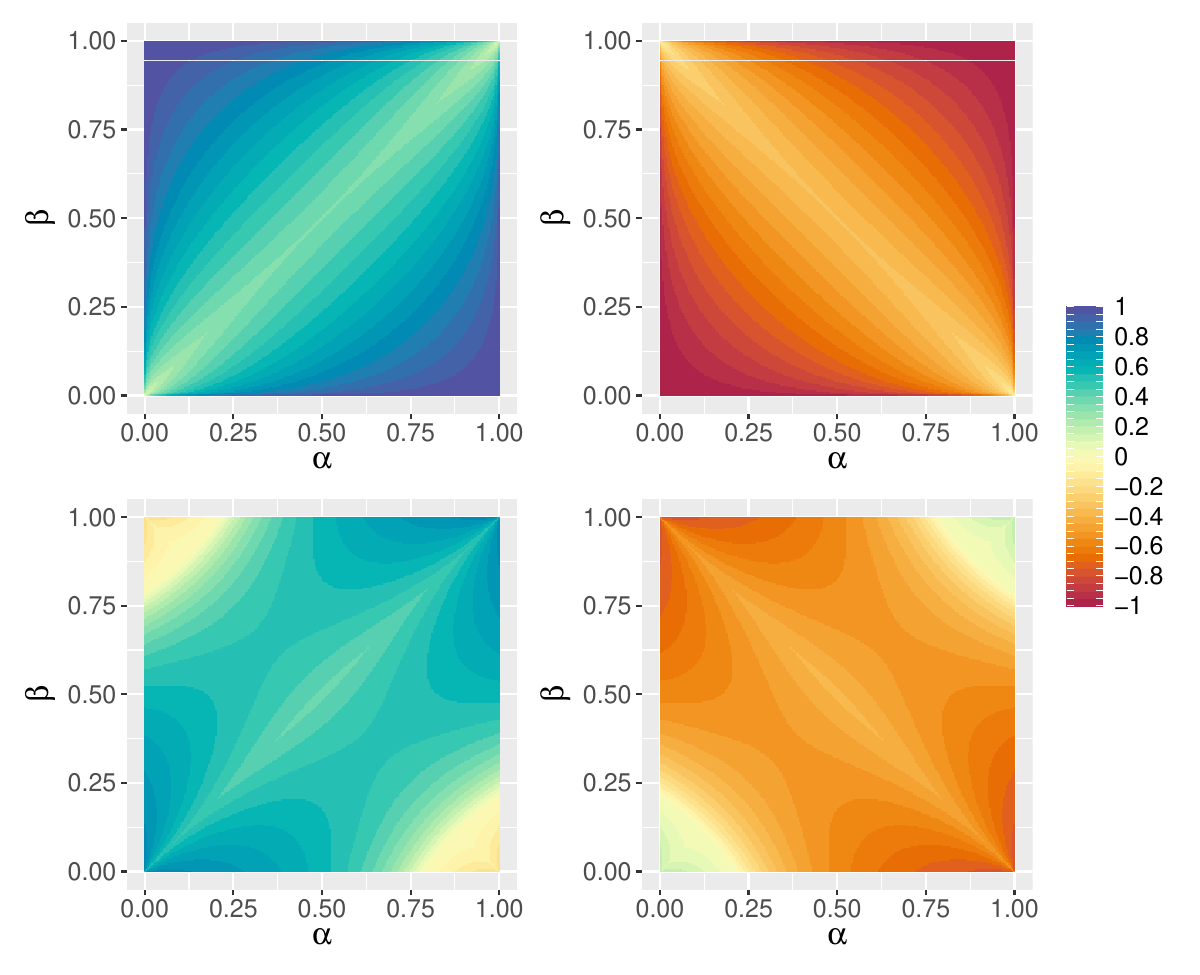}
	\caption{Plots of $\QFCor$ for four different copulas, Gaussian (upper left) and Cauchy (lower left) with Spearman's $\rho$ equal to 0.5: Gaussian (upper right) and Cauchy (lower right) with Spearman's $\rho$ equal to -0.5}
	\label{fig:qfcorvariousnegative}
\end{figure}

\begin{figure}[ht]
	\centering
	\includegraphics[width=1\linewidth]{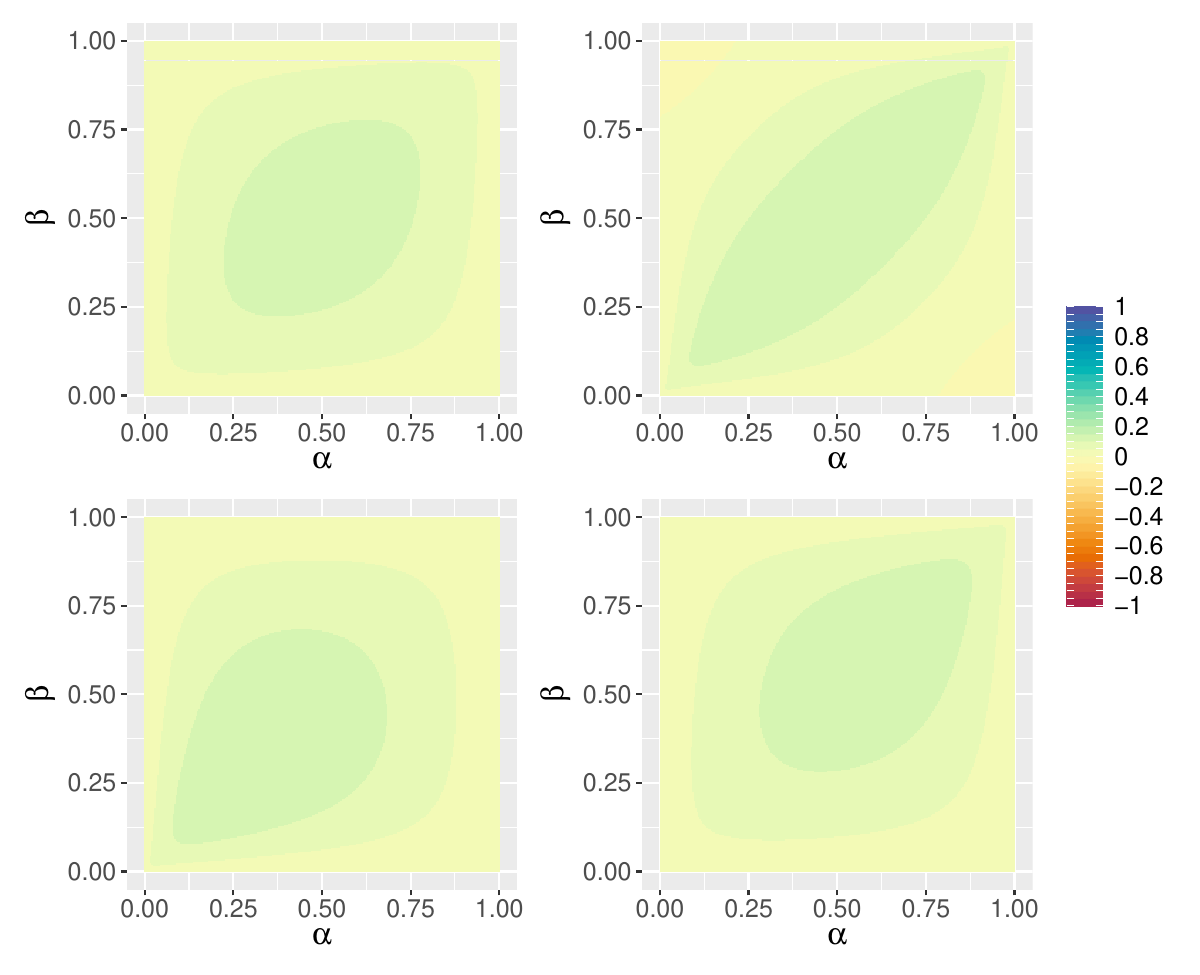}
	\caption{Plot of $\QFCov$ normalised via the Cauchy--Schwarz inequality, see \eqref{eq:CSI 3}, for four different copulas, all with Spearman's $\rho$ equal to 0.5: Gaussian (upper left), Cauchy (upper right), Clayton (lower left) and Gumbel (lower right)}
	\label{fig:qfcorvariousCS}
\end{figure}

\section{Additional material for the data examples}	
\label{app:data examples}

\begin{figure}[ht]
	\centering
	\includegraphics[width=0.6\linewidth]{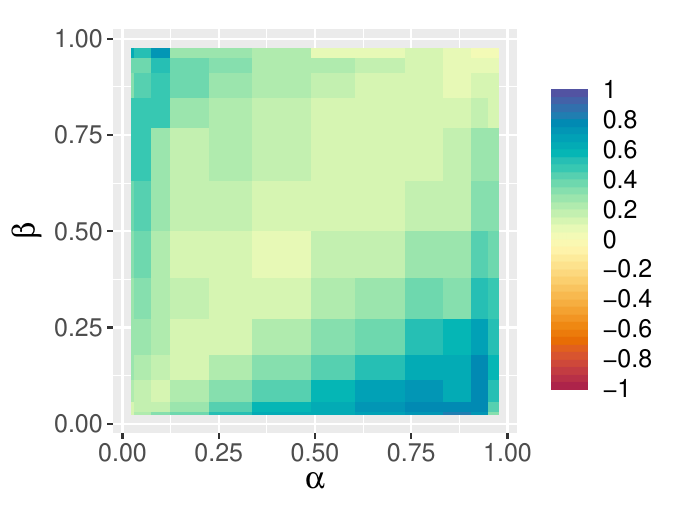}
	\caption{$\QFCor$ of heights of mixed-sex couples}
	\label{fig:qfcorheights}
\end{figure}

\begin{figure}[ht]
	\centering
	\includegraphics[width=1\linewidth]{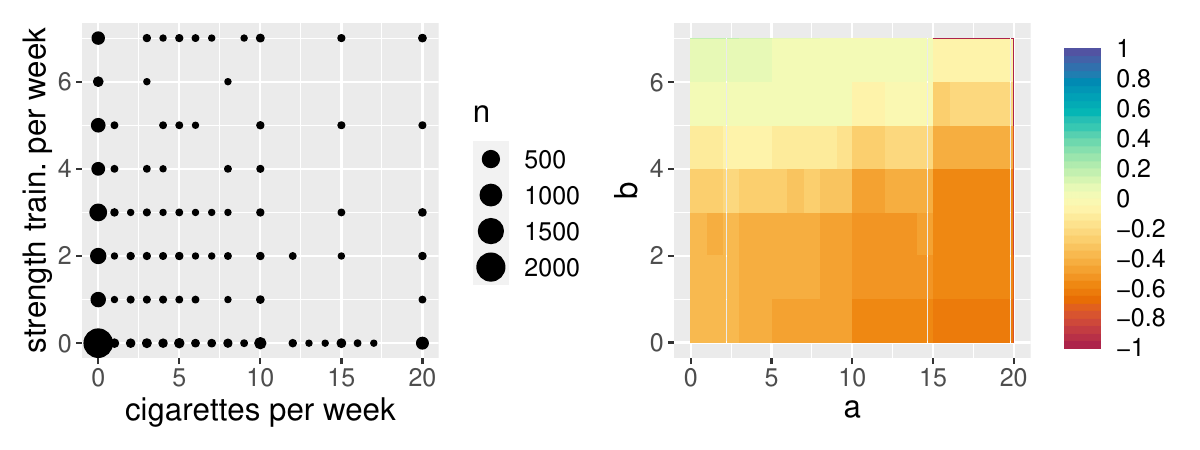}
	\caption{Bubble plot and empirical $\CDFCor$ of number of cigarettes smoked per week and number of strength trainings per week of the men}
	\label{fig:cdfcorcigarettesstrengthtrainingplusscatter}
\end{figure}

We briefly consider the relation between Pearson correlation from \eqref{eq:Pearson} and mean correlation from Example \ref{ex:Mean correlation}. 
Table \ref{tab:MCorvsPearsonCor} contains those two quantities, their ratio and the skewness of the two variables considered for the three examples above and a fourth example. The skewness serves here as a rough proxy for how ``close'' the marginal distributions of the two variables are to fulfilling the conditions of Lemma \ref{lem:symmetric}. For the first two examples mean and Pearson correlation are virtually identical as the distributions of heights of men and women are close to being normal and the distributions of BMIs are right-skewed, but have a very similar shape: Thus, in both cases the distributions are close to being of the same type and this implies that under positive dependence the Fr\'echet--Hoeffding and the Cauchy--Schwarz normalisation are almost identical. For the weekly frequency of strength training and BMI of men covariance is negative and the marginal distributions of one variable and the negative of the other are skewed in opposite directions, leading to a substantial difference between the two normalisations and a ratio of 0.71 between the two. As a fourth, more extreme, example we consider the number of cigarettes smoked per week and the number of strength training per week for the men in the sample. Figure \ref{fig:cdfcorcigarettesstrengthtrainingplusscatter} presents  bubble plot and $\CDFCor$ for this example. Again, covariance is negative and both variables are heavily right-skewed, leading to the distributions of one and the negative of the other variable being ``far away'' from being of the same type and to Pearson correlation substantially understating strength of dependence: $\Cor=-0.087$, whereas $\MCor=-0.400$.

\begin{table}
	\caption{$\Cor$, $\MCor$, their ratio, skewness of the two variables for our data examples}
	\centering
	\begin{tabular}{@{}ccccccc@{}}
		\toprule
		$X$ & $Y$  & $\Cor$ & $\MCor$ & $\Cor/\MCor$ & $\mathrm{Skew}(X)$ & $\mathrm{Skew}(Y)$ \\ \midrule
		height woman & height man & $\hphantom{-}0.213$ & $\hphantom{-}0.216$ & $\hphantom{-}0.99$ & $-0.09$ & $-0.11$  \\
		BMI woman & BMI man  & $\hphantom{-}0.303$ & $\hphantom{-}0.304$ &  $\hphantom{-}1.00$ & $\hphantom{-}1.14$ & $\hphantom{-}1.12$ \\
		freq.\ training & BMI & $-0.088$ & $-0.122$ & $\hphantom{-}0.71$ & $\hphantom{-}1.51$ & $\hphantom{-}1.12$   \\
		freq.\ smoking & freq.\ training  & $-0.087$ & $-0.400$ & $\hphantom{-}0.22$ & $\hphantom{-}4.17$ & $\hphantom{-}1.51$ \\
		\bottomrule 
	\end{tabular}
	\label{tab:MCorvsPearsonCor}
\end{table}


\end{document}